\documentclass[onefignum,onetabnum]{siamonline181217}

\usepackage{lipsum}
\usepackage{amsfonts}
\usepackage{graphicx}
\usepackage{subcaption}
\usepackage{graphbox}
\usepackage{epstopdf}
\usepackage{algorithmic}
\usepackage{booktabs}
\usepackage{dsfont}
\usepackage[margin=1.25in]{geometry}
\usepackage[colorinlistoftodos]{todonotes}

\ifpdf
  \DeclareGraphicsExtensions{.eps,.pdf,.png,.jpg}
\else
  \DeclareGraphicsExtensions{.eps}
\fi

\usepackage{enumitem}
\setlist[enumerate]{leftmargin=.5in}
\setlist[itemize]{leftmargin=.5in}

\newsiamremark{remark}{Remark}
\newsiamremark{hypothesis}{Hypothesis}
\crefname{hypothesis}{Hypothesis}{Hypotheses}
\newsiamthm{claim}{Claim}
\newcommand{\mcA}{{\mathcal{A}}}
\newcommand{\Ff}{\mathcal{F}}

\newcommand{\FF}{\mathbb{F}}
\newcommand{\PP}{\mathbb{P}}
\newcommand{\EE}{\mathbb{E}}

\newcommand{\mfT}{{\mathfrak{T}}}
\newcommand{\tT}{{t\in\mfT}}

\newcommand{\mfG}{{\mathfrak{G}}}

\newcommand{\Id}{{\mathds 1}}

\graphicspath{{figures/}}

\headers{Optimal Behaviour in SREC Markets}
{A. Shrivats and S. Jaimungal}

\title{Optimal Generation and Trading in Solar Renewable Energy Certificate (SREC) Markets\thanks{SJ would like to acknowledge the support of the Natural Sciences and Engineering Research Council of Canada
(NSERC), [funding reference numbers RGPIN-2018-05705 and RGPAS-2018-522715]}\\\vspace{0.5cm}{\color{black}\normalfont {\large{\emph{Forthcoming in Applied Mathematical Finance}}}}}

\author{Arvind Shrivats \and Sebastian Jaimungal\thanks{Department of Statistical Sciences, University of Toronto, Toronto, ON
  (\email{shrivats@utstat.toronto.edu},\email{sebastian.jaimungal@utoronto.ca} ).}}

\usepackage{amsopn}

\ifpdf
\hypersetup{
  pdftitle={Optimal Behaviour of Regulated Firms in Solar Renewable Energy Certificate (SREC) Markets},
  pdfauthor={Arvind Shrivats, Sebastian Jaimungal}
}
\fi

\begin{document}

\maketitle
\begin{abstract}
SREC markets are a relatively novel market-based system to incentivize the production of energy from solar means. A regulator imposes a floor on the amount of energy each regulated firm must generate from solar power in a given period and provides them with certificates for each generated MWh. Firms offset these certificates against the floor and pay a penalty for any lacking certificates. Certificates are tradable assets, allowing firms to purchase/sell them freely. In this work, we formulate a stochastic control problem for generating and trading in SREC markets from a regulated firm's perspective. We account for generation and trading costs,  the impact both have on SREC prices, provide a characterization of the optimal strategy, and develop a numerical algorithm to solve this control problem. Through numerical experiments, we explore how a firm who acts optimally behaves under various conditions. We find that an optimal firm's generation and trading behaviour can be separated into various regimes, based on the marginal benefit of obtaining an additional SREC, and validate our theoretical characterization of the optimal strategy. We also conduct parameter sensitivity experiments. 
\end{abstract}



\section{Introduction}
As the impacts of climate change continue to be felt worldwide, policies to reduce greenhouse gas emissions and promote renewable energy generation are of increasing importance. One approach that encapsulates many policies is market-based solutions. The most well-known of the policies which fall under this umbrella are carbon cap-and-trade (C\&T) markets.

In carbon C\&T markets, regulators impose a limit on the amount of carbon dioxide ($\text{CO}_2$) that regulated firms can emit during a certain time period (referred to as a compliance period). They also distribute allowances (credits) to individual firms in the amount of this limit, each allowing for a unit of $\text{CO}_2$ emission, usually one tonne. Firms must offset each of their units of emissions with an allowance, or face a monetary penalty for each allowance they are lacking. These allowances are tradable assets, allowing firms who require more credits than what they were allocated to buy them, and firms who require less to sell them. In this way, C\&T markets aim to find an efficient way of allocating the costs of $\text{CO}_2$ abatement across the regulated firms.

In practice, these systems regulate multiple consecutive and disjoint compliance periods, which are linked together through mechanisms such as \textit{banking}, where unused allowances in period-$n$ can be carried over to period-$(n+1)$. Other linking mechanisms include \textit{borrowing} from future periods (where a firm may reduce its allotment of allowances in period-$(n+1)$ in order to use them in period-$n$) and \textit{withdrawal}, where non-compliance in period-$n$ reduces period-$(n+1)$ allowances by the amount of non-compliance (in addition to the monetary penalty previously mentioned).

A closely related alternative to these cap-and-trade markets are \textit{renewable energy certificate} markets (REC markets). A regulator sets a floor on the amount of energy generated from renewable sources for each firm (based on a percentage of their total energy generation), and provides certificates for each MWh of energy produced via these means\footnote{Not all generators of renewable energy who participate in REC markets are regulated Load Serving Entities (LSEs), though in this work, we largely focus on the decisions faced by those who are regulated.}. \textcolor{black}{This is also known as a Renewable Portfolio Standard (RPS).} To ensure compliance, each firm must surrender certificates totaling the floor at the end of each compliance period, with a monetary penalty paid for each lacking certificate. The certificates are traded assets, allowing regulated Load Serving Entities (LSEs) to make a choice about whether to produce electricity from renewable means themselves, or purchase the certificates on the market (or a mix of both). 

REC markets can be used to encourage growth of a particular type of renewable energy. The most notable of these systems are Solar REC markets (SREC markets), which have been implemented in many areas of the northeastern United States\footnote{The largest and most mature SREC market in North America is the New Jersey SREC Market}, and are the focus of this work.

The similarities between carbon cap-and-trade markets and SREC markets are clear. However, there are also some notable differences. One key difference between the SREC market and traditional carbon cap-and-trade markets is the uncertainty in the former market is the supply of certificates (driven by some generation process), while in the latter, the uncertainty is in the demand for allowances (driven by an emissions process). In SREC markets, banking is typically implemented, but borrowing and withdrawal are not. Broadly speaking, SREC markets can be considered the inverse of a cap-and-trade system.

The existing literature on SREC markets largely focus on certificate price formation. \cite{coulon_khazaei_powell_2015} presents a stochastic model for SREC generation. They also calibrate it to the New Jersey SREC market, and ultimately solve for the certificate price as a function of economy-wide generation capacity and banked SRECs, and investigates the role and impact of regulatory parameters on these markets. The volatility of REC prices has been noted in other works, such as \cite{amundsen2006price} and \cite{hustveit2017tradable}. The latter  focuses on the Swedish-Norwegian electricity certificate market and develops a stochastic model to analyze price dynamics and policy. \cite{khazaei2017adapt} studies an alternate design scheme for SREC markets and shows how it can stabilize SREC prices. 

Additionally, there are extensive studies of the  carbon cap-and-trade markets, particularly in developing stochastic equilibrium models for emissions markets. \cite{hitzemann2018equilibrium} presents a general stochastic framework for firm behaviour leading to the expression of allowance price as a strip of European binary options written on economy-wide emissions. Agents' optimal strategies and properties of allowance prices are also studied by \cite{carmona2010market} and \cite{seifert_uhrig-homburg_wagner_2008} within a single compliance period setup, with the former also making significant contributions through detailed analyses of potential shortcomings of these markets and their alternatives. \cite{carmona_fehr_hinz_2009} also proposes a stochastic equilibrium model to explain allowance price formation and develop a model where abatement (switching from less green to more green fuel sources) costs are stochastic. There is also significant work on structural models for financial instruments in emissions markets, such as \cite{howison_schwarz_2012} and \cite{carmona_coulon_schwarz_2012}.

Our contribution addresses a natural question in these systems; how should regulated LSEs behave? Here, we use stochastic control techniques to  characterize firm specific optimal behaviour through  generation and trading and discuss potential takeaways from a market design perspective. We believe these results are of interest to both regulators, the designers of SREC markets (\textcolor{black}{and REC markets in general}), and the firms regulated by them.

Specifically, we explore a cost minimization problem of a single regulated firm in a single-period SREC market with the goal of understanding their optimal behaviour as a function of their current level of compliance and the market price of SRECs. To this end, we  pose the problem as a continuous time stochastic control problem. We provide the optimality conditions, and analyze the form of the optimal controls in feedback form to illuminate features of the solution. In addition, we numerically solve for the optimal controls of the regulated firm  as generation and trading costs vary, including a detailed analysis of various scenarios and sample paths. We also explore the sensitivity of the optimal controls to the various parameters in the model. We extend these results to a single regulated firm in a multi-period SREC market.

There are several differences between our work and the extant literature. Firstly, we focus on the SREC market, which is a new and burgeoning market and there are few studies (in comparison to carbon C\&T markets). Secondly, we focus on the optimal behaviour of firms, something that has not been studied in SREC markets. In the carbon literature, prior works formulate a stochastic control problem in order to better understand the behaviour of the allowance prices, while we begin with an SREC price process (which regulated agents affect by trading and generation) and are interested in how the agent should optimally behave. We assume  that agents affect the SREC price process in a manner similar to the permanent price impact models  in the optimal execution literature  (see \cite{almgren2001optimal}, \cite{cartea2015algorithmic}).

The remainder of this work is organized as follows. \textcolor{black}{\Cref{srec_markets} provides a background on REC markets in practice and SREC markets in particular.} \Cref{model} discusses our model and poses the general optimal behaviour problem in continuous time.  \Cref{optimality} presents optimality results in a continuous time setting. \Cref{discrete} provides a discrete time formulation and numerically solves the dynamic programming equation to characterize the optimal behaviour of a regulated firm. Finally, in \Cref{results}, we present the results of our work including sensitivity analysis.

\section{SREC Market Overview} \label{srec_markets}

\textcolor{black}{RPS regulations have been instituted in numerous regions around the globe. In this section, we provide a brief overview of their use, with a particular focus on RPS regulations in the United States. These regulations aim to promote the production of electricity via solar energy (among potentially other energies) through the use of SREC markets. While we focus on the United States RPS regulations and their associated REC / SREC markets in this section, we note that RPS regulations have also been instituted around the globe, including China, Sweden, and Norway, among others. }

\textcolor{black}{Roughly 30 US states have enacted RPS regulations, see \cite{kolesnikoff_cleveland_shields_2019}. These regulations typically apply to investor owned utilities (IOUs) which  are private LSEs (as opposed to municipal or state LSEs). Such IOUs supply electricity to the grid as part of their business-as-usual operations. They receive a (tradable) REC for each MWh of electricity they generate from renewable means. The RPS requires that  regulated LSEs submit RECs annually in an amount proportional to their total electricity supply. They face a monetary penalty for any RECs they lack under the amount required by the RPS.}

\textcolor{black}{REC markets also include players who are not regulated by the RPS, but may have the ability to produce RECs. Often these are individuals who have attached solar panels to their place of residence,  registered  with the tracking authority, and sell the resulting RECs on the market.}

\textcolor{black}{RPS regulations can be stratified further. Many of the 30 US states that have RPS' have also instituted a `carve-out' for solar energy -- that is, a specification that a certain proportion of the renewable energy generated must come from solar means. This results in an SREC market, as opposed to a general REC market -- here, the certificates represent the solar nature of the generated energy. As such, the IOUs in such states must specifically generate solar energy (or purchase SRECs from an IOU / individual that has) in order to comply with the RPS regulation that applies to them. Unused certificates can be banked for a given amount of years before expiring \footnote{In New Jersey's SREC market, unused SRECs can be banked for four additional years, giving them a five-year life in total.}. This results in different `vintages' of SRECs in the market, depending on the year the SREC was produced (as that will impact how many years it can continue being banked into the future). In this work, we make a simplifying assumption that firms can bank SRECs indefinitely to avoid dealing with multiple vintages of SRECs and as including multiple vintages does not add more insight into the problem.
}

\textcolor{black}{The solar carve out is typically not large, relative to the overall distribution of electricity generation. New Jersey's SREC market, the largest and most mature in North America, has a solar carve-out of just 5.1\% of overall electricity sales in 2021, after which the state is transitioning to a new, currently undetermined solar energy generation incentive (see \cite{lane_2020}). In other states such as Washington D.C. (who are continuing their SREC programs for the foreseeable future), the solar carve out is planned to reach 10\% of overall electricity sales by 2041. }

\textcolor{black}{As New Jersey's market is the most mature of the North American SREC markets, we discuss it in more detail. While there is a dearth of academic literature regarding these markets, useful background information regarding it can be found on \cite{nj_power}. Despite the winding down of the New Jersey SREC market, discussion of it is nonetheless useful in order to better understand what a relatively  mature SREC market looks like. In general, SREC markets still figure to be an important component of energy policy, with states like Maryland and Washington D.C. (as alluded to above) continuing to develop their own solar carve outs and associated SREC markets. Additionally, REC markets in general will  continue to grow in importance in the future, and our work  applies to REC markets in general.}

\textcolor{black}{Now focusing specifically on the NJ SREC market, we plot the number of SRECs issued in \cref{fig:issued_SRECs}. We retrieve this data from PJM-GATS, the administrator which tracks the New Jersey SREC market (see \cite{pjm}). The figure shows consistent increases in generated SRECs, as well as the seasonality effect. This latter property is natural due to reduced sunlight in winter months. As each SREC corresponds to a MWh of electricity generated from solar means, the figure suggests that market's monthly solar generation nears 40,000 MWh at its peak (around June 2019). From this, we can see the notable growth of the NJ SREC market. Over this time, the NJ SREC market has undergone numerous (significant) regulatory changes. These included changes to the requirement schedule, the penalty schedule, as well as the rules around banking of unused SRECs. The most notable of these changes occurred in 2012, where the regulatory body drastically decreased the non-compliance penalty, increased the SREC requirement, and allowed for extended banking of unused SRECs. We do not discuss these changes further in this section, except to remark that they did occur, and have contributed to the observed patterns.}

\begin{figure}[!t]
\begin{minipage}[t]{0.31\textwidth}
    \centering
    \includegraphics[align = c, width=\textwidth]{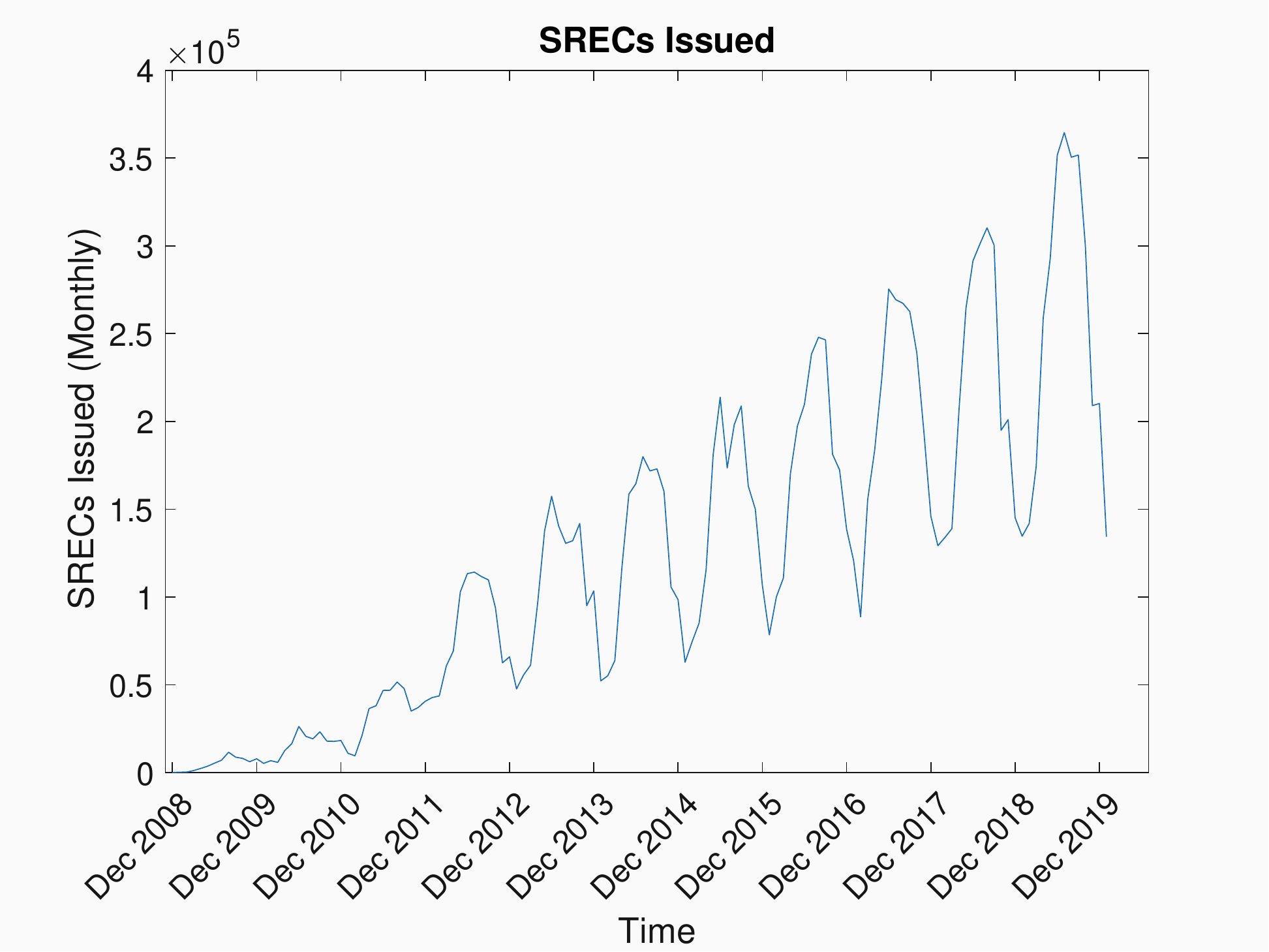}
    \caption{Issued SRECs in New Jersey SREC Market from 2008 - 2019}
    \label{fig:issued_SRECs}
\end{minipage}
\hfill
\begin{minipage}[t]{0.31\textwidth}
    \centering
    \includegraphics[align = c, width=\textwidth]{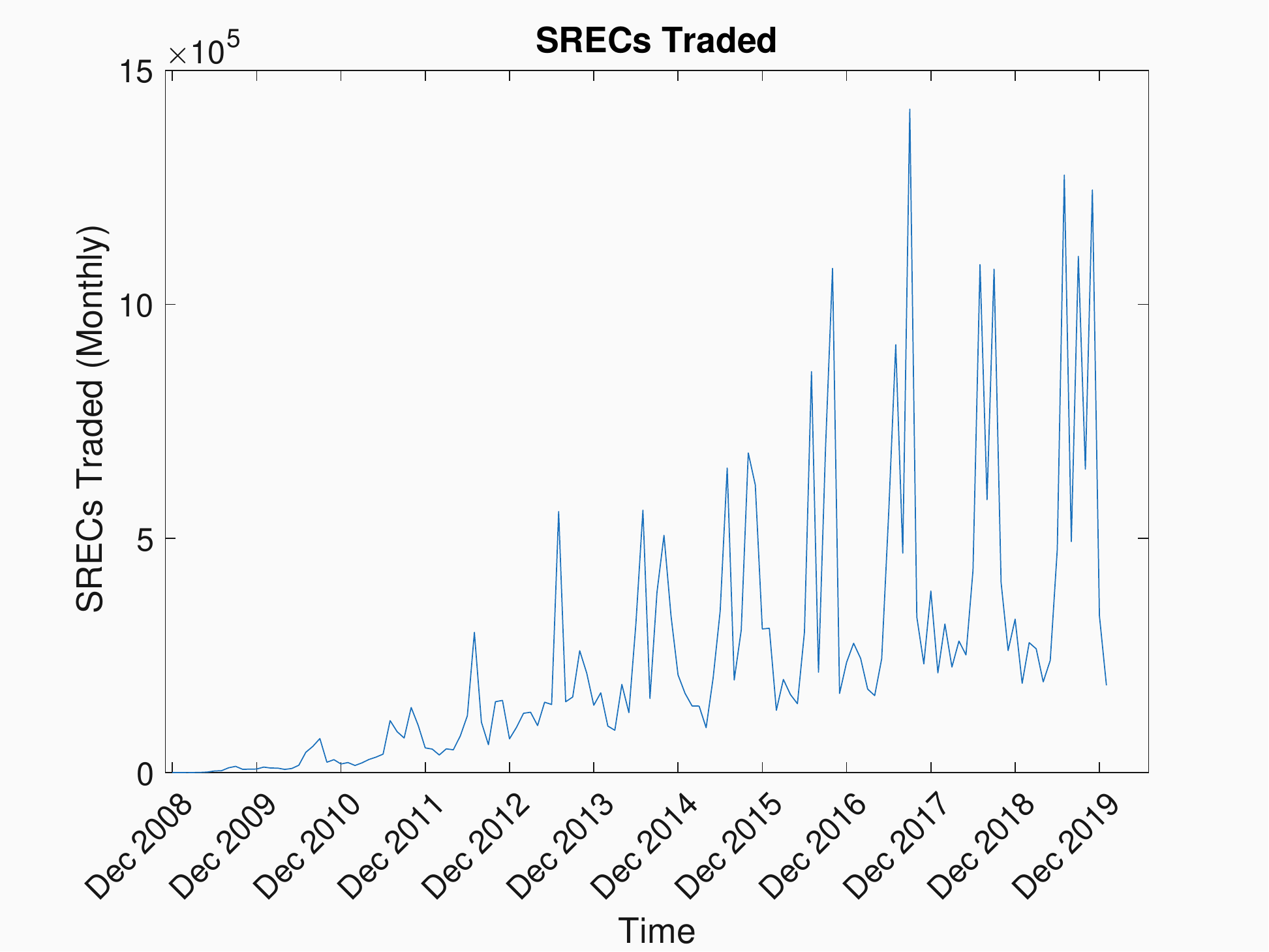}
    \caption{Traded SRECs (all vintages) in New Jersey SREC Market from 2008 - 2019}
    \label{fig:traded_SRECs}
\end{minipage}
\hfill
\begin{minipage}[t]{0.31\textwidth}
    \centering
    \includegraphics[align = c, width=\textwidth]{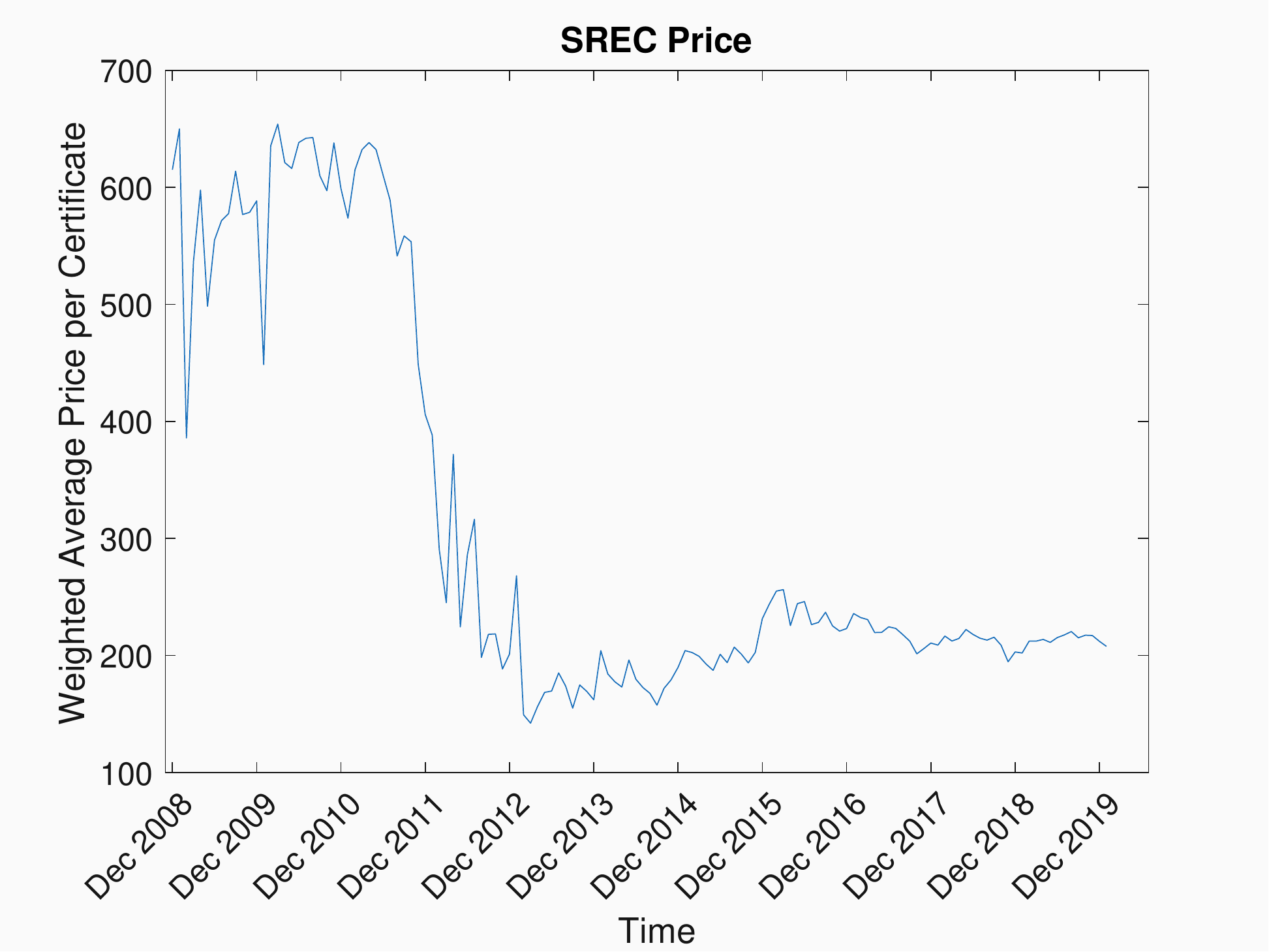}
    \caption{Weighted average SREC price (across all vintages) in New Jersey SREC Market from 2008 - 2019}
    \label{fig:srec_price}
\end{minipage}
\end{figure}
\textcolor{black}{This growth is also apparent when we plot the magnitude of trading activity (of all SREC vintages) over time, as in \cref{fig:traded_SRECs}. Once again, we see notable seasonality, with the peaks occurring around June and October of each year; this has to do with the compliance dates of the energy year in the NJ SREC market (the energy year runs from June 1 - May 31), along with a `true-up' period. This is a period of six months, from June 1 to Nov 30, which is the time span firms have between the end of the energy year and the compliance date when they must submit their SRECs.}

\textcolor{black}{Next, we plot the weighted average monthly SREC price from 2008 onward, in \cref{fig:srec_price}. This is the average price of each SREC vintage sold in each month, weighted by the relative proportion of each vintage in the market. In practice, older vintages will typically trade at a slight discount to the current vintage, as they can be banked for less time in the future, on account of being older. We plot the weighted average for visual simplicity, because the difference between prices of older vintages and the current vintage tends to be rather small, and because the most recent vintage tends to have higher trading volumes.}

\textcolor{black}{From \cref{fig:srec_price}, we see a large amount of variance in the price in the early stages of the SREC market. The NJ SREC market initially sustained high SREC prices (close to the non-compliance penalty). However, as time passed, increased investment into solar generation led to oversupply, leading to lower prices. In conjunction with the aforementioned regulatory changes that lowered non-compliance penalties and increased the SREC requirement, this resulted in the notable drop in SREC price which occurred throughout 2011 and 2012. Since then, price changes have been less dramatic. For further information on the New Jersey SREC market and its price history, please see \cite{coulon_khazaei_powell_2015}, \cite{khazaei2017adapt}.}

\textcolor{black}{Finally, we remark on the number of agents regulated by SREC markets. As mentioned earlier, all electricity suppliers (equivalently, LSEs) in the licensed within the region over which the market is enforced must comply with the RPS obligations. In New Jersey, this includes hundreds of such firms (see \cite{supplier_list_2020} and \cite{nj_power}), in addition to the various other players who do not have RPS obligations, but have solar facilities that allow them to produce and sell SRECs.}

\section{The SREC Generation and Price Impact Model} \label{model}

\subsection{SREC Market Rules}

We assume the following rules for the SREC market, which are exogenously specified and fixed. In an $n$-period framework, a firm is obliged to submit $(R_1, ..., R_n)$ SRECs at the end of the compliance periods $[0, T_1], ..., [T_{n-1}, T_n]$, respectively. \textcolor{black}{As discussed in the previous section, the requirement that regulated firms are subject to is based on a proportion of the electricity they supply to the grid. We make a simplifying assumption, similar to \cite{coulon_khazaei_powell_2015}, that this requirement is instead exogenous.}

For the period $[T_{i-1}, T_i]$, firms  pay $P_i$ for each SREC below $R_i$ at $T_i$. Firms receive an SREC for each MWh of electricity they produce through solar energy. We assume firms may bank leftover SRECs not needed for compliance into the next period, with no expiry on SRECs. This is a simplifying assumption we make -- many SREC markets have limitations on how long an SREC can be banked for (e.g., in New Jersey's SREC market, an SREC can be banked for a maximum of four years). This assumption reduces the dimensionality of the state space. After $T_n$, all SRECs are forfeited.

A single period framework follows the rules above with $n = 1$. For convenience, we remove the subscripts in the notation for the terms defined above when discussing a single-period framework. That is, the regulated firm is required to submit $R$ SRECs at time $T$, representing their required production for the compliance period $[0, T]$. A penalty $P$ is imposed for each missing SREC at time $T$. The firm considers any costs/profits arising from the SREC system after $T$ to be immaterial. 

\subsection{Firm Behaviours}

We first consider a single firm who is optimizing their behaviour in a single compliance period SREC system. A regulated firm can control their planned generation rate (SRECs/year) at any given time $(g_t)_{\tT}$ (where $\mfT:=[0,T]$) and their trading rate (SRECs/year) at any given time $(\Gamma_t)_{\tT}$. The processes $g$ and  $\Gamma$ constitute the firm's controls. 

The trading rate may be positive or negative, reflecting that firms can either buy or sell SRECs at the prevailing market rate for SRECs. Firms also incur a trading penalty of $\frac{1}{2}\gamma \Gamma_t^2$, $\gamma > 0$, per unit time. This induces a constraint on their trading speed. In general, the quadratic penalty could be replaced by any convex function of $\Gamma_t$.

In an arbitrary time period $[t_1, t_2]$, The firm aims to generate $\int_{t_1}^{t_2} g_t \,dt$, but actually generates $\int_{t_1}^{t_2} g_t^{(r)} dt = \int_{t_1}^{t_2} g_t dt + \int_{t_1}^{t_2} \nu_t dB_t^{(1)}$, where $\nu_t$ is a deterministic function of time, and $\nu_tdB_t^{(1)}$ may be interpreted as the generation rate uncertainty at $t$. We assume that a firm has a baseline deterministic generation level $h_t$ (SRECs/year), below which there is no cost of generation. Methods similar to \cite{coulon_khazaei_powell_2015} may be used to estimate $h_t$. We assume that $h_t < \infty$ for all $t$. Increases in planned generation from their baseline production incurs the cost $C(g, h) := \frac{1}{2}\zeta (g - h)_+^2$ per unit time. 

\textcolor{black}{This choice of cost $C$ may be viewed as a `rental' cost for temporarily increasing SREC generation capacity, as opposed to an `investment' cost. It is as though the firm is renting additional capacity (i.e. solar panels) with which they are able to increase their SREC generation. However, this additional cost does not result in any long-run increases to their baseline production $h_t$. It is certainly possible to model instead expansion efforts rather than `rental' costs, but we leave such extensions for future work. 
}

In \cite{aid2017coordination}, the authors utilize a similar cost structure in the context of expanding solar capacity. There, costs are quadratic, but not one-sided. Lastly, we note our choice of $C$ is both differentiable and convex; any choice of $C$ with these properties could be used instead in order for the analysis contained in this paper to be valid. 

All processes are defined on the filtered probability space $(\Omega, \mathcal{F}, \FF = (\Ff_t)_{t\geq0}, \PP)$, where $\FF$ is the natural filtration generated by the SREC price. The set of admissible controls $\mcA$ equals the set of all progressively measurable (with respect to $\FF$) processes $(g_t,\Gamma_t)_\tT$ such that $\EE[\int_0^T g_t^2\, dt ] < \infty$, $\EE[\int_0^T\Gamma_t^2\,dt] < \infty$, and $g_t \geq 0$ for all $\tT$. 
At time $t$ the firm holds $b_t^{g, \Gamma}$ SRECs and the (controlled) SREC price is denoted $S_t^{g, \Gamma}$. 

The various processes satisfy the stochastic differential equations (SDEs)
\begin{subequations}
\begin{align}
S_t^{g, \Gamma} &= S_0 + \int_0^t\left( \mu_u + \eta\, \Gamma_u  - \psi \,g_u \right) du - \int_0^t \psi\, \nu_u\, dB_u^{(1)} + \int_0^t  \sigma_u\, dB_u^{(2)} , \quad \text{and}   \label{eq:S} \\
b_t^{g, \Gamma} &= b_0 + \int_0^t (g_u + \Gamma_u)\, du + \int_0^t \nu_u dB_u^{(1)},
\end{align}%
\label{eqn:S-and-b-SDE}
\end{subequations}%
where $B=(B^{(1)}_t,B^{(2)}_t)_{t\ge0}$ is a standard two-dimensional Brownian Motion
and $\mu, \sigma, \nu$ are deterministic functions. We further assume $\int_0^T \sigma_u^2 du < \infty$ and $\int_0^T \nu_u^2 du < \infty$. As the SDE above indicates,  trading ($\int \Gamma_u du$) and realized generation ($\int g_u^{(r)} du$) impact the SREC price linearly. As such, our model is similar to the price impact models commonly studied in optimal execution problems. Buying (selling) of SRECs pushes the price up (down) and generation
pushes the price downwards. In this way, a firm's behaviour impacts the rest of the market. SREC inventory ($b_t^{g, \Gamma}$) accumulates  by both trading and generation activity. 

\textcolor{black}{We observe that this exogenous specification of $S_t^{g, \Gamma}$ does not necessarily converge to any particular value as $t \rightarrow T$. A hallmark of many equilibrium pricing models in both the C\&T and SREC literature is that the price of the certificate (allowance) converges to either the non-compliance penalty or $0$, depending on the compliance of the economy as a whole (see \cite{coulon_khazaei_powell_2015},  \cite{hitzemann2018equilibrium}, \cite{seifert_uhrig-homburg_wagner_2008}, among others). However, this work takes a different approach. While the value of an SREC to an \textit{individual} firm at time $T$ is $0$ (if the firm has complied) or $P$ (if the firm has not complied), this does not imply that the price that the market bears will necessarily be one of these two. This is because the SREC (or indeed, C\&T) market comprises of \textit{many} regulated firms, some of whom will find the certificate valueless, and some of whom will find it worth $P$. The relative proportions of agents who have complied or failed to comply will inform the certificate price and render it in the interval $[0,P]$ rather than $0$ or $P$ only.} 

For any admissible strategy $g,\Gamma\in\mcA$, a  regulated firm's performance criterion (at time $t$) for the single-period problem is 
\begin{multline}
J^{g, \Gamma}(t, b, S) = -\EE_{t,b,S}\left[\int_t^T C(g_u, h_u) \,du  \right.
\\
+ \int_t^T \Gamma_u\, S_u^{g, \Gamma} \,du
+ \left. \tfrac{\gamma}{2} \int_t^T \Gamma_u^2\, du 
+ P (R - b_T^{g, \Gamma})_+\right], \label{eq:pc}
\end{multline}
where   $\EE_{t,b,S}\left[\cdot\right]$ denotes taking expectation conditioned on $b_t = b$ and $S_t=S$ and in the sequel, $\EE_t[\cdot]:=\EE[\cdot|\mathcal{F}_t]$ and $\PP_t[\cdot]:=\PP[\cdot|\mathcal{F}_t]$.
The firm's cost minimization is the strategy which attains the sup (if it exists) below and the value of the optimal strategy is
\begin{equation}
V(t, b, S) = \sup_{(g_s, \Gamma_s)_{s\in[t,T]} \in \mcA} J^{g,\Gamma}(t, b, S). \label{eq:vf}
\end{equation}

In the next section, we characterize the optimal trading strategy and the relationship to SREC price using the stochastic maximum principle as well as the dynamic programming equation approach. 
 
\section{Continuous time approach} \label{optimality}

\subsection{Stochastic Maximum Approach}

One approach to solving \eqref{eq:vf} is through 
the Stochastic (Pontryagin) Maximum Principle (see the seminal works of \cite{kushner1972necessary} and \cite{peng1990general}). Here, we apply the stochastic maximum principle to our  problem along the lines of \cite{hitzemann2018equilibrium}. In doing so, we characterize the optimal controls as a system of coupled equations. The key result is contained in the following proposition.

\begin{proposition}[Optimality Conditions]\label{opt_cond}
The processes $(g,\Gamma)=(g_t, \Gamma_t)_{\tT}$  satisfying
the forward-backward stochastic differential equations (FBSDEs)
\begin{align}
\Gamma_t &= \tfrac{1}{\gamma}\left(M_t - S_0 - \int_0^t (\mu_u + \psi \,g_u)\, du\right),\label{eq:eq_gam}
\\
\Gamma_T &= \tfrac{1}{\gamma} \left(P\, \Id_{\{b_T^{g,\Gamma} < R\}} - S_T^{g,\Gamma}\right),
\\
g_t &= \left(h_t + \tfrac{1}{\zeta}\left(Z_t - \psi \int_0^t \Gamma_u du\right)\right)\Id_{\left\{P\, \PP_t(b_T^{g,\Gamma} < R) \geq - \psi \,\EE_t[\int_t^T \Gamma_u du]\right\}}\,,\label{eq:eq_g}
\\
g_T &= \tfrac{1}{\zeta}\left(P\, \Id_{\left\{b_T^{g,\Gamma} < R\right\}} + \zeta\, h_t\right),
\end{align}
for all $\tT$, where the processes $(M,Z)=(M_t,Z_t)_{\tT}$ are martingales, are the  optimal controls for problem \eqref{eq:vf}.
\end{proposition}
\begin{proof}

The Hamiltonian for the performance criterion \eqref{eq:pc} and state dynamics \eqref{eqn:S-and-b-SDE} is
\begin{multline}
\mathcal{H}(t, b, S, g, \Gamma, \boldsymbol{y}, \boldsymbol{z}) =
- \tfrac{\zeta}{2} ((g - h_t)_+)^2 - S \Gamma - \tfrac{\gamma}{2} \Gamma^2 
\\
+ y_{b}\, (g + \Gamma) + y_{S} (\mu_t + \eta \Gamma - \psi g)  + \sigma_t z_S - \psi \nu_t z_{S,b} + \nu_t z_b,
\label{eq:hamiltonian}
\end{multline}
where $\boldsymbol{y}=(y_b,y_S), \boldsymbol{z} =
\begin{bmatrix}
    z_b       & z_{bS} \\
    z_{Sb}       & z_S
\end{bmatrix}$.

This is  concave in the controls $g,\Gamma$ and state variables $b,S$. Moreover, the adjoint processes $(y_b, y_S)=(y_{b,t},y_{S,t})_{\tT}$ satisfy the BSDEs
\begin{subequations}
\label{eqn:yb-and-ys-SDE}
\begin{align}
dy_{b, t} &= z_{b, t}\, dB_t^{(1)} + z_{bS, t}\, dB_t^{(2)},
&
y_{b, T} = P\, \Id_{\{b_T^{g,\Gamma} < R\}}.
\\
dy_{S, t} &= \Gamma_t \,dt + z_{S,t}\, dB_t^{(2)} + z_{Sb, t}dB_t^{(1)}, 
&
y_{S, T} = 0.
\end{align}
\end{subequations}

The stochastic maximum principle implies that if there exists a solution $(\boldsymbol{\hat{y}}, \boldsymbol{\hat{z}})$ to  \eqref{eqn:yb-and-ys-SDE}, then a strategy $(g, \Gamma)$ that maximizes $\mathcal{H}(t, b, S, g, \Gamma, \boldsymbol{\hat{y}}, \boldsymbol{\hat{z}})$ is the optimal control.

As both BSDEs have linear drivers, their solution is  straightforward  (see \cite{pham2009continuous}, Chapter 6) and given by
\begin{equation}
y_{b, t} = P\, \PP_t(b_T^{g,\Gamma} < R)\,, 
\qquad \text{and}
\qquad
y_{S, t} = -\EE_t\left[\int_{t}^T \Gamma_u du\right]\,.
\label{eqn:y-adjoint-sol}
\end{equation}

Differentiating the Hamiltonian with respect to the controls, we obtain the first order conditions
\begin{subequations}
\begin{align}
\frac{\partial \mathcal{H}}{\partial \Gamma} &: \quad y_{b} + \eta\, y_{S} - S - \gamma\, \Gamma = 0,\qquad \text{and} \\
\frac{\partial \mathcal{H}}{\partial g} &: \quad y_{b} - \psi\, y_{S} - \zeta\, (g - h_t)_+ = 0,
\end{align}%
\end{subequations}%
and substituting the solutions to the adjoint processes \eqref{eqn:y-adjoint-sol}, we obtain the  optimality conditions
\begin{subequations}
\begin{align}
P\,\PP_t(b_T^{g,\Gamma} < R) - \eta\,\EE_t\left[\int_{t}^T \Gamma_u du\right] - S_t^{g,\Gamma} - \Gamma_t \,\gamma &= 0, \qquad \text{and}
\label{eq:optGam}
\\
P \,\PP_t(b_T^{g,\Gamma} < R) + \psi\, \EE_t\left[\int_{t}^T \Gamma_u du\right] - \zeta\, (g_t - h_t)_+ &= 0.
\label{eq:optg}
\end{align}
\end{subequations}
We next, aim to solve these equations by isolating $g$ and $\Gamma$. 

First, from \eqref{eq:optGam} we have
\begin{equation}
Y_t + \eta \int_0^t \Gamma_u du - S_t^{g,\Gamma} = \Gamma_t \gamma, 
\label{eqn:GammaBSDE-step1}
\end{equation}
where $Y=(Y_t)_{\tT}$ is the Doob-martingale defined by
\begin{equation}
Y_t = P\, \PP_t(b_T^{g,\Gamma} < R) - \eta\, \EE_t\left[\int_0^T \Gamma_u du\right]\,.
\end{equation}

Rearranging \eqref{eqn:GammaBSDE-step1} and substituting in \eqref{eq:S}, we arrive at \eqref{eq:eq_gam} where the terminal condition  follows immediately from  \eqref{eq:optGam}, and $M= (M_t)_{t \in [0, T]}$ is the martingale defined by
\begin{equation}
M_t = Y_t - \int_0^t \sigma_u dB_u^{(2)} + \psi \int_0^t \nu_u dBu^{(1)}\,.
\end{equation}

For \eqref{eq:optg}, consider a modification $\mathcal{A}^U$ of the set of admissible  controls $\mathcal{A}$ to controls that admit a finite upper bound $U > \sup_{\tT} h_t$\footnote{Any  bound that is lower is practically meaningless as firms must be able to generate at or more than their `baseline' generation rate.}.


When $g_t \geq h_t$, the solution to \eqref{eq:optg} is
\begin{equation}
    g_t = h_t + \tfrac{1}{\gamma}K_t, \qquad
\text{where}
\qquad
K_t = P\,\PP_t(b_T < R) + \psi\,\EE_t\left[\int_t^T \Gamma_u du\right]
\label{eq:g_soln} 
\end{equation}
Define $g_t^\star:=h_t + \tfrac{1}{\gamma}K_t$. 

When, $g_t<h_t$, the Hamiltonian is maximized at
\begin{equation}
    g_t = \begin{cases} 
   U, & \text{if } K_t \geq 0, \\
   0,       & \text{otherwise}.
  \end{cases} \label{eq:g_star2}
\end{equation}
Denote the sets
\begin{equation*}
    A := \left\{g_t \ge h_t\right\},  \quad
    B:=\left\{\{g_t < h_t\} \cap \left\{ K_t \geq 0\right\}\right\}, \quad \text{and}
    \quad
    C:= \left\{\{g_t < h_t\} \cap \left\{ K_t < 0\right\}\right\}.
\end{equation*}%
These sets satisfy the property that $A = (B \cup C)^c$. From \eqref{eq:g_soln} and \eqref{eq:g_star2}, the optimal generation rate is therefore
\begin{align}
    g_t = g_t^\star\, \Id_{A}+ U\, \Id_{B}.
    \label{g_t_star}
\end{align}
Consider the set $A^\star = \{g_t^\star \geq h_t\}$.
\begin{lemma}
If 
$U>\sup_{\tT}~h_t$, then $A^\star = A$.
\end{lemma}
\begin{proof}

Take an event $\omega\in A$, by \eqref{g_t_star} $g_t(\omega) = g_t^\star(\omega)$ and so $g_t^\star(\omega) \ge h_t$, and hence $\omega\in A^\star$. Therefore $A \subset A^\star$.


Take an event $\omega\in A^\star$, so that $g_t^\star(\omega)\ge h_t$. As
$g_t(\omega) = g_t^\star(\omega) \Id_{\{\omega\in A\}}+ U\, \Id_{\{\omega\in B\}}$ and $U \geq\sup_{\tT} h_t$,  we must have that $g_t(\omega) \geq h_t$, and thus $\omega\in A$. Therefore, $A^\star \subset A$.
\end{proof}

Therefore, we can rewrite \eqref{g_t_star} as follows:
\begin{equation}
    g_t = g_t^\star\, \Id_{A^\star} + U\, \Id_{B}.
    \label{g_t_star_upd}
\end{equation}
Furthermore, $B=\emptyset$ because, from \eqref{g_t_star_upd}, $\omega\in B \implies g_t(\omega) = U > h_t$, so $\omega\notin B$. 

Therefore, we obtain
\begin{equation}
    g_t = \begin{cases} 
   h_t + \frac{1}{\gamma}K_t, & \text{if } K_t \ge 0, \\
   0       & \text{otherwise}.
  \end{cases} \label{eq:g_eqn}
\end{equation}

On the set ${K_t\ge 0}$, by adding and subtracting $\psi \int_0^t \Gamma_u du$ to $g_t$ and letting $Z=(Z_t)_{\tT}$ be the Doob-martingale defined by 
\begin{equation}
Z_t = P \,\PP_t(b_T < R) + \psi\, \EE_t\left[\int_{0}^T \Gamma_u du\right]
\end{equation}
we obtain \eqref{eq:eq_g} with the terminal condition obtained from \eqref{eq:optg}. As (i) this solution is independent of $U$ for all $U>\sup_{\tT}h_t$, (ii) $\sup_{\tT}h_t<\infty$,
and (iii) $\mathcal{A}=\lim_{U\to\infty}\mathcal{A}^U$, this completes the proof.
\end{proof}
We end this subsection with a few comments regarding the results of this proposition and interpretations of the optimality conditions.
In comparison to \cite{hitzemann2018equilibrium}, where the authors develop optimality conditions in a carbon C\&T system, our results shows that the trading penalty and the impact of trading and generation on SREC prices modifies the optimality conditions. When $\eta = \psi = 0$, \eqref{eq:optg} reduces to $P\, \PP_t(b_T < R) = \zeta (g_t - h_t)_+$. This is similar to the result that the marginal cost of generation is equal to the product of the penalty and probability of non-compliance  found in \cite{hitzemann2018equilibrium}. Moreover, when $\eta = \psi = 0$, \eqref{eq:optGam} reduces to $P \,\PP_t(b_T < R) - \gamma \Gamma_t  = S_t$. Thus, in our setup, the  SREC price equals the  penalty scaled by the probability of non-compliance but modified by the optimal trading of the firm.

Similar behavior persists in the general case when $\eta > 0, \psi > 0$. From \eqref{eq:optGam}, the  SREC price equals the penalty scaled by the probability of non-compliance, but modified by the time-$t$ marginal cost of the firm's trading and our expectations of their future trading. That is, low prices are  associated with high rate of trading and high expected future rate of trading.

From \eqref{eq:optg}, the  penalty scaled by the probability of non-compliance equals the difference between the marginal cost of generation and re-scaled (by $\psi$) expected future trading.

As well, from the form  of $g_t$ in \eqref{eq:g_eqn}, the firm  plans to either generate above their baseline or not at all. The decision is contingent on the inequality $P \PP_t(b_T < R) \geq - \psi \EE_t[\int_t^T \Gamma_u du]$. If the inequality is satisfied, then the firm will  generate above their baseline, and if not, they cease to generate. Intuitively, this condition represents whether the firm benefits enough from generation to offset the potentially negative influence of their price impacts. The term $P \PP_t(b_T < R)$ represents the expected non-compliance cost avoided by acquiring an additional SREC, while $- \psi \EE_t[\int_t^T \Gamma_u du]$ represents the value lost through the impact of generating an additional SREC. Generation puts downwards pressure on $S$ (through \eqref{eq:S}), which the firm realizes through their expected trading level over the remainder of the compliance period. Note that $- \psi \EE_t[\int_t^T \Gamma_u du]$ is only positive if expected future trading is negative - that is, the firm expects to be a seller of SRECs.

We outline two simple examples to demonstrate the effect of this property. Consider a firm that is far from compliance, and thus, $\EE_t[\int_t^T \Gamma_u du] > 0$ (that is, the firm expects to purchase SRECs during $[t, T]$). The indicator in \eqref{eq:g_eqn} (and consequently, \eqref{eq:eq_g}) is always satisfied, and the firm will generate above the baseline $h$. This is consistent with the behaviour of a firm that has not reached compliance and is striving to acquire enough SRECs to hit the requirement $R$. Conversely, consider a firm that has SRECs well in excess of $R$ and plans to sell them over the remainder of the compliance period. Here, the indicator is not satisfied. That is, an additional generated SREC will not help the firm's compliance probability significantly (if at all), and additional generation will  decrease the SREC price $S$, reducing the revenue generated by the firm through sales. As such, the firm chooses not to generate at all in such a scenario, and in doing so, mitigates the  negative effect generation has on SREC prices. 

We can solve Equations \eqref{eq:eq_gam}-\eqref{eq:eq_g}  numerically using Least Square Monte Carlo techniques. Instead, we consider a dynamic programming  approach to solving the original problem   \eqref{eq:vf}.

\section{Discrete time version of problem} \label{discrete}

Thus far, we formulated the cost minimization problem of a single regulated firm using continuous time optimal control techniques to characterize the solution and tease out some essential features of the optimal strategy. To obtain numerical solutions, however, we solve a discrete time version of the problem which we find has better numerical stability. Indeed, a discrete time formulation more closely approximates practice, as regulated firms typically take actions only at discrete  time points within a compliance period.

To this end, let $n$ be the number of decision points within a single compliance period,
which occur at $0 = t_1 < t_2 < ... < t_n < T = t_{n+1}$. For simplicity, we assume these are equally spaced so that $t_k=k\Delta t$.

The control processes $(g, \Gamma)$ are now piecewise constant within $[t_{i},t_{i+1})$, and the firm controls $\{g_{t_i}, \Gamma_{t_i}\}_{ i \in \mathfrak{N}}$ where $ \mathfrak{N}:=\{0,\dots,n\}$, so that at each time point, the regulated firm chooses their trading and generating behaviour over the next interval of length $\Delta t$. In this section, $(g, \Gamma)$ represent  vectors whose elements are these controls.

Under the same assumptions as earlier, the performance criterion (corresponding to the total cost) for an arbitrary admissible control is 
\begin{align}
J^{g, \Gamma}(m, b, S)
= \EE_{t_m,b,S}\left[  \sum_{i=m}^n \left\{ \tfrac{\zeta}{2} ((g_{t_i} - h_{t_i})_+)^2   + \Gamma_{t_i} S_{t_i}^{g, \Gamma} + \tfrac{\gamma}{2} \Gamma_{t_i}^2\right\} \Delta t  + P (R - b_T^{g, \Gamma})_+\right],
\label{eq:discretepc}
\end{align}
In the above, the dynamics of the state variables ($b, S$) are modified for discrete time to
\begin{subequations}
\begin{align}
S_{t_i}^{g, \Gamma} &= \min\left(\left(S_{t_{i-1}}^{g, \Gamma} + \left(\mu  + \eta \,\Gamma_{t_{i-1}}  - \psi\, g_{t_{i-1}}\right) \Delta t - \psi \nu \sqrt{\Delta t}\, \varepsilon_{t_i} + \sigma \sqrt{\Delta t}\, Z_{t_i}\right)_+,\; P\right) \label{eq:6}\\
b_{t_i}^{g, \Gamma} &= b_{t_{i-1}}^{g, \Gamma} +  (g_{t_{i-1}} + \Gamma_{t_{i-1}})\Delta t + \nu \sqrt{\Delta t} \varepsilon_{t_i} \label{eq:7}
\end{align}%
\label{eqn:DiscreteStateEvolution}
\end{subequations}%
where $Z_{t_i}, \varepsilon_{t_i} \sim N(0, 1)$, iid, for all $i \in \mathfrak{N}$.

Note that \ref{eq:6} is the discrete time analogue of \eqref{eq:S}  capped at $P$ and floored at $0$. The cap and floor ensures that SREC prices remain in the closed interval $[0,P]$ as prices outside this interval cannot occur in real markets.

We aim to optimize \eqref{eq:discretepc} with respect to $(g, \Gamma)$ and determine the value of the position of the regulated firm, as well as their optimal behaviour. Hence, we seek
\begin{align}
V(t, b, S) = \inf_{g, \Gamma\in\mcA} J^{g, \Gamma}(t, b, S), \label{eq:DiscreteValueFunc}
\end{align}
and the strategy that attains the inf, if it exists. Applying the Bellman Principle to \eqref{eq:DiscreteValueFunc} implies
\begin{subequations}
\begin{align}
\begin{split}
V(t_i, b, S) &= \inf_{g_{t_i}, \Gamma_{t_i}} \biggl\{
\left(\tfrac{\zeta}{2}  ((g_{t_i} - h_{t_i})_+)^2
+ \Gamma_{t_i} S_{t_i}^{g, \Gamma}
+ \tfrac{\gamma}{2} \Gamma_{t_i}^2\right) \Delta t
\\ & \qquad\qquad\quad
+ \EE_{t_i}\left[V\left(t_{i+1}, b_{t_{i+1}}^{g, \Gamma}, S_{t_{i+1}}^{g, \Gamma}\right)\right] \biggr\}, \qquad\qquad \text{and}
\label{eq:Discrete-Bellman}%
\end{split}%
\\
V(T, b, S) &= P (R - b)_+\,. \label{eq:Discrete-Bellman-terminal}
\end{align}%
\label{eqn:DiscreteBellman}%
\end{subequations}%
%
In the next section, we  provide a numerical scheme for solving this optimization problem.

\section{Solution Algorithm and Results} \label{results}

\textcolor{black}{The remainder of the paper is contained in this section and we briefly pause here to provide an overview of its contents.}

\textcolor{black}{We begin by discussing the parameter choices for the numerical experiments we  run. In the best scenario, we would  calibrate our model to real-world data, where possible. This is difficult for several reasons. First,  our framework requires the costs that regulated firms experience in the SREC market, through their planned generation and trading activities. Such information is, however, only known to industry insiders who trade and generate their  SRECs. Instead, for the parameters that do not have clear real-world values to tether to, we choose what we feel are sensible values that allow us to understand the general behaviours and principles that govern this dynamical system.} 

\textcolor{black}{Following the parameter choice, we present our algorithm to solve \eqref{eq:DiscreteValueFunc} in detail. From here, we transition into the presentation of the numerical simulations themselves. Naturally, there are many possible interesting scenarios to simulate. We present only a fraction of what is possible and theoretically interesting here, and leave to the reader to explore others. Instead, we select a handful of interesting scenarios that reveal important relationships between  the firm's optimal behaviour and the state processes, and discuss them in depth. These include an overall summary of the regulated firm's optimal controls as a function of both state variables, simulated compliance periods, summary statistics, and sensitivity analysis. The code for this project can be found \href{https://github.com/AShrivats/SREC_Single_Player.git}{here}.}


\subsection{Parameter Choice}

For the first set of numerical experiments, we use the parameters reported in Tables \ref{tbl:ComplianceParams} and \ref{tbl:ModelParams}.
\begin{table}[h]
\begin{minipage}[c]{0.54\textwidth}
\begin{center}
 \begin{tabular}{ccccc}
\toprule\toprule
 $n$& $T$ &$P$ (\textdollar/SREC) &$R$& ${h_t}$ (SREC/y) \\
 \midrule
 50 & 1&  300 & 500 & 500 \\
\bottomrule\bottomrule
\end{tabular}
\caption{Compliance parameters.
\label{tbl:ComplianceParams}}
\end{center}
\end{minipage}
\begin{minipage}[c]{0.45\textwidth}
\begin{center}
 \begin{tabular}{ccccccc}
\toprule\toprule
 $\mu$ & $\sigma$ & $\nu$ & $\psi$ &$\eta$ &$\zeta$ & $\gamma$ \\
 \midrule
 0 & 10 & 10 & 0.01 & 0.01 & 0.6 & 0.6 \\
\bottomrule\bottomrule
\end{tabular}
\caption{Model Parameters. \label{tbl:ModelParams}}
\end{center}
\end{minipage}
\end{table}

These parameters are chosen for illustrative purposes. As discussed in the preamble of this section, calibration to a particular firm  is itself a non-trivial problem and requires proprietary knowledge of a firm's cost function and baseline production (which  also varies significantly from firm to firm). Instead, we provide broad-level intuition regarding the optimal behaviour of a firm in a single-period SREC market with reasonable parameters. The penalty of $P = \$300$ is informed by the New Jersey SREC market, where the non-compliance penalty in compliance period ending May 2018 is $\$308$ \textcolor{black}{(see \cite{SRECTrade})}. \textcolor{black}{In practice, as the requirement $R$ is based on a proportion of sales for each regulated firm, the specific level of $R$ should be tied to sales. We choose $R$, however, to be exogenous as opposed to a stochastic process depending on electricity sales in order to simplify the analysis (this assumption is present in other works, such as \cite{coulon_khazaei_powell_2015}).} The choice  ${h_t} = \frac{R}{T}$ implies the regulated firm has a probability of $0.5$ to comply if they simply plan to generate at their baseline rate and do not partake in the SREC market. 

The values of $\zeta$ and $\gamma$ are motivated by the upper bounds they imply for $g_t, \Gamma_t$. Specifically, consider the case of a firm that cannot generate enough solar energy to meet the requirements, and hence will fail to comply. The benefit of generating SRECs is to reduce their non-compliance obligation, and with each generated SREC their  obligation is reduced by $P$. Therefore, the costs and benefits of generation over a time-step are (independent of trading activity), respectively,
\begin{equation}
    K_1(g_t) = \tfrac{1}{2}\zeta((g_t - {h_t})_+)^2\Delta t, \qquad\text{and} \qquad
    B_1(g_t) = P g_t \Delta t\,.
\end{equation}

The firm generates energy in order to minimize $N_1(g_t) := K_1(g_t) - B_1(g_t)$ which occurs at $g_t^* = \tfrac{P}{\zeta} + {h_t}$. For the chosen parameters, $g^* = 1,000$ which is exactly twice the baseline rate ${h_t}$. In other words, this choice of $\zeta$ ensures the firm's maximum generation rate is bounded by twice their baseline.

We conduct a similar exercise for $\Gamma_t$. Consider a firm that will fail to comply.
In this scenario, a rational firm will purchase SRECs. As before, the benefit of a firm purchasing SRECs is to reduce their non-compliance obligation, with each generated SREC reducing the obligation by $P$. As such, the costs and benefits to purchase over the next time-step are (independent of generation activity):
\begin{equation}
    K_2(\Gamma_t) = \left(\tfrac{1}{2}\gamma\,\Gamma_t^2 + S_t\, \Gamma_t\right)\Delta t,
    \qquad\text{and}\qquad
    B_2(\Gamma_t) = P\, \Gamma_t\, \Delta t,
\end{equation}
respectively. The firm purchases in order to minimize $N_2(\Gamma_t) := K_2(\Gamma_t) - B_2(\Gamma_t)$ which occurs at $\Gamma_t^* = \tfrac{P - S}{\gamma}$. For the chosen parameters, this is maximized when $S=0$ and results in $\Gamma^* = 500$. The significance of this computation is to show we have chosen parameters that result in a reasonable upper bound on the amount of trading a firm will partake in.

Repeating the same exercise for a firm that is guaranteed to comply (and thus is motivated to sell), we  obtain $g_t^* = 0$ (due to price impacts of generation) and $\Gamma_t = -\tfrac{S}{\gamma}$ which is maximized (in absolute value) at $-500$ for the chosen parameters.

For the parameters in Tables \ref{tbl:ComplianceParams} and \ref{tbl:ModelParams}, this simple analysis shows that generation and trading rates are restricted to the range $g_t \in [0, 1000]$ and $\Gamma_t \in [-500, 500]$, which is a reasonable range of possible values given our choices of ${h_t}$ and $R$.

We set $\eta = 0.01, \psi = 0.01$ to demonstrate the effect of price impact. We justify these choices in a similar manner to the above. While $\eta$ and $\psi$ do not need to be equal in our model, generation and purchasing are substitutes for one another, and thus, it is logical to consider them as having equal market impacts. Above, we discussed natural bounds for the control processes, given the compliance and model parameters in $\cref{tbl:ComplianceParams}$ and $\cref{tbl:ModelParams}$. For these parameter choices, price impact parameters $\eta = 0.01, \psi = 0.01$ results in upper bounds for net price impacts of $500 \times 0.01 \times \tfrac{1}{50} = 0.1$ per time-step for trading and $1,000 \times \ 0.01 \times \tfrac{1}{50} = 0.2$ per time-step for generation. These are sizeable impacts for a single firm to have, but not so large that the model seems implausible.


In \Cref{sensitivity}, we consider other parameters. In particular, we explore how various levels of $\zeta, \gamma$ impact firm behaviour and the effect of other price impact parameters ($\psi \neq 0, \eta \neq 0$). In the following subsection, we detail our algorithm to solve the dynamic programming problem outlined in \Cref{discrete}.

\subsection{Numerical Scheme} \label{num_scheme}
We use the following numerical algorithm for solving \eqref{eqn:DiscreteBellman} with state variable dynamics in \eqref{eqn:DiscreteStateEvolution}:
\begin{enumerate}

\item Choose a grid of $b$ and $S$ values denoted by $\mfG$. We use a uniform grid of  $401$ points in $b$ from $0$ to $2R$, so that $R$ is on the grid, and a uniform grid of $S$ with $\Delta S = \sqrt{3 \Delta t} \sigma$ and lower and upper bounds of $0$ and $P$ respectively. In this manner, the number of grid points in $S$ is tuned to the volatility over a time-step\footnote{As with any numerical solution, there is a trade-off between grid size (accuracy of the dynamic program solution) and  run-time. The grid we use provides an acceptable trade-off between these two, and we observed no further increase in accuracy by increasing the grid size.}

\item Minimize \eqref{eq:Discrete-Bellman} at $i=n$ (corresponding to $t=T-\Delta t$) with respect to $(g_{t_n}, \Gamma_{t_n})$ for every point in $\mfG$.

    To do this, we require an estimate of  $\EE_{t_n}\left[V\left(t_{n+1}, b_{t_{n+1}}^{g_{t_n}, \Gamma_{t_n}}, S_{t_{n+1}}^{g_{t_n}, \Gamma_{t_n}}\right)\right]$
for each 

$(S_{t_n},b_{t_n})\in\mfG$. This is achieved by simulation as follows:
    \begin{enumerate}[label=\emph{\Alph*.}]
\item Select a value $b_{t_n}\in\mfG$. As the terminal condition is independent of $S_{t_{n+1}}^{g_{t_n}, \Gamma_{t_n}}$, the optimal controls and value function for $b_{t_n}$ will be the same for all values of $S_{t_n}$. That is, the evolution of the SREC price is unimportant at the last time-step.

    \begin{enumerate}
    \item Select a candidate pair $(g_{t_n},\Gamma_{t_n})$
    \begin{enumerate}[label=\emph{(\alph*)}]
        \item Simulate $100$ scenarios of $b_{t_{i+1}}^{g_{t_i}, \Gamma_{t_n}}$ using \eqref{eq:7}, -- use the same set of random numbers for all points in $\mfG$.

        \item For each simulated $b_{t_{n+1}}^{g_{t_n},\Gamma_{t_n}}\,$, calculate the one-step-ahead 
        
        value function $V\left(t_{n+1}, b_{t_{n+1}}^{g_{t_i},\Gamma_{t_n}},S_{t_{n+1}}^{g_{t_n},\Gamma_{t_n}}\right)$ through the terminal condition \eqref{eq:Discrete-Bellman-terminal}

        \item Use the empirical mean of the result of (c) as an estimate of the true mean at $b_{t_n}$, which will be the same regardless of $S_{t_n}$.
    \end{enumerate}
    \item Use Matlab's \texttt{fmincon} function to determine next candidate pair $(g_{t_i},\Gamma_{t_i})$ and repeat from (i) until converged, store optimal pair and value function.
    \end{enumerate}

    \item Go to next grid point in $\mfG$ repeat from A.
\end{enumerate}

\item Step backwards from $i+1$ to $i$, by minimizing \eqref{eq:Discrete-Bellman} with respect to $(g_{t_i}, \Gamma_{t_i})$ at time $t_i$ for all points in $\mfG$.

To do this, we require an estimate of  $\EE_{t_i}\left[V\left(t_{i+1}, b_{t_{i+1}}^{g_{t_i}, \Gamma_{t_i}}, S_{t_{i+1}}^{g_{t_i}, \Gamma_{t_i}}\right)\right]$
for each 

$(S_{t_i},b_{t_i})\in\mfG$. This is achieved by simulation as follows:
\begin{enumerate}[label=\emph{\Alph*.}]
\item Select a pair $(S_{t_i},b_{t_i})\in\mfG$

    \begin{enumerate}
    \item Select a candidate pair $(g_{t_i},\Gamma_{t_i})$
    \begin{enumerate}[label=\emph{(\alph*)}]
        \item Simulate $b_{t_{i+1}}^{g_{t_i}, \Gamma_{t_i}}$ using \eqref{eq:7}, -- use the same set of random numbers for all points in $\mfG$.

        \item Simulate $100$ scenarios  of $S_{t_{i+1}}^{g_{t_i},\Gamma_{t_i}}$ by applying \eqref{eq:6} -- use the same set of random numbers for all points in $\mfG$. 

        \item For each simulated pair of $(b_{t_{i+1}}^{g_{t_i},\Gamma_{t_i}}\,,\,S_{t_{i+1}}^{g_{t_i},\Gamma_{t_i}})$, estimate the one-step-ahead value function $V\left(t_{i+1}, b_{t_{i+1}}^{g_{t_i},\Gamma_{t_i}},S_{t_{i+1}}^{g_{t_i},\Gamma_{t_i}}\right)$ by interpolation.

        \item Use the empirical mean of the result of (c) as an estimate of the true mean at $(b_{t_i},S_{t_i})$.
    \end{enumerate}

    \item Use Matlab's \texttt{fmincon} function to determine next candidate pair $(g_{t_i},\Gamma_{t_i})$ and repeat from (i) until converged, store optimal pair and value function.
    \end{enumerate}

    \item Go to next grid point in $\mfG$ repeat from A.
\end{enumerate}

\end{enumerate}

This procedure provides an estimate of the value function at all grid points $\mfG$ and at all times $\mathfrak{T}:=\{t_i\}_{i\in\mathfrak{N}}$, as well as the optimal generation and trading rates on $\mfG\times\mathfrak{T}$.

\textcolor{black}{In the following subsections, we apply this scheme in a variety of simulation studies to learn more about the optimal behaviours of the regulated firms, and their associated implications.}

\subsection{Optimal Behaviours of a Regulated Firm}

A regulated firm's optimal behaviour is one of the key outputs from solving the Bellman equation. \Cref{fig:opt_beh} shows the dependence of the optimal trading and generation rate on banked SRECs for three SREC prices at six points in time. 
\begin{figure}[!htp]
\centering
\includegraphics[align = c, width=0.6\textwidth]{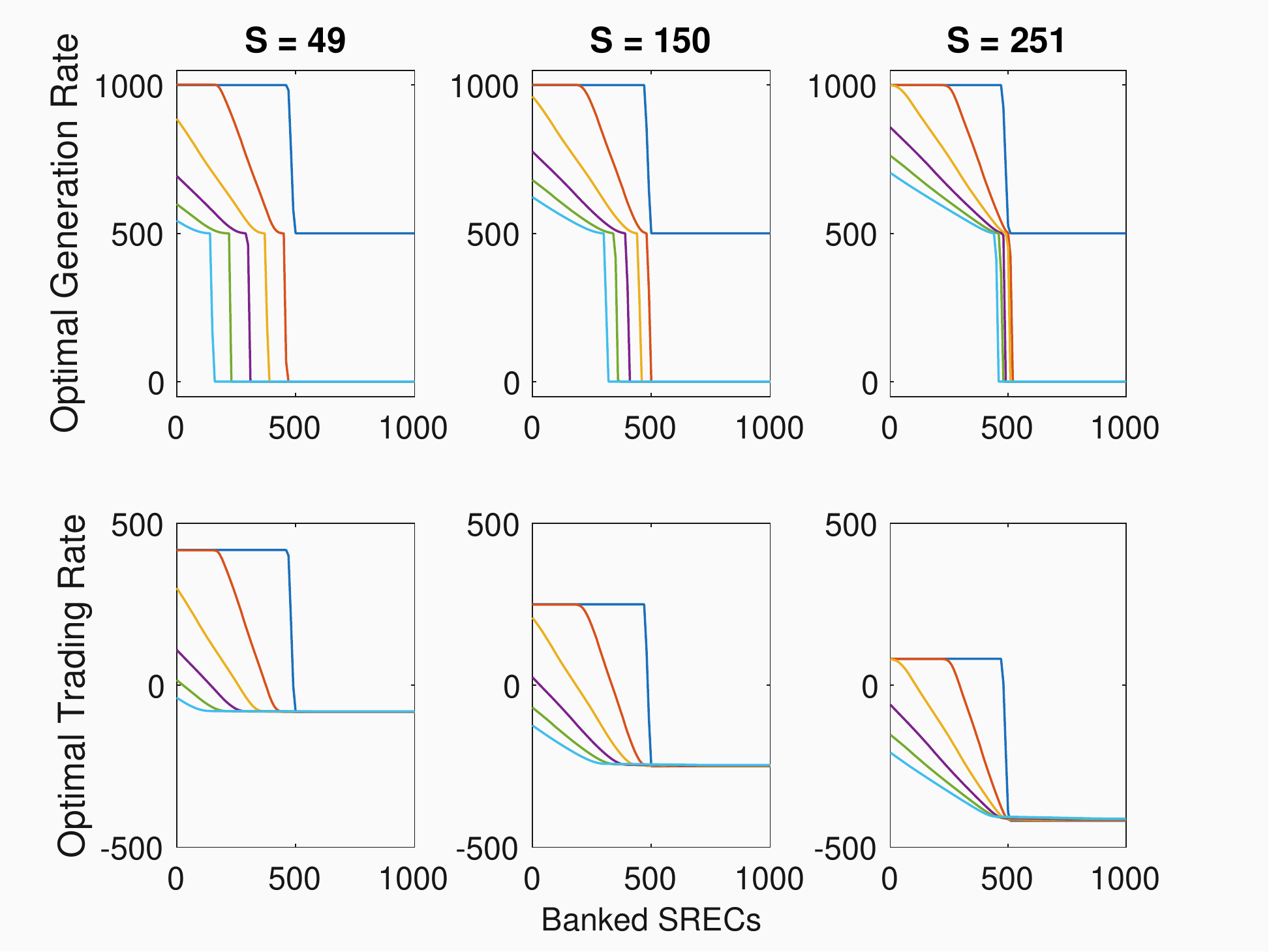}
\includegraphics[align = c, width=0.1\textwidth]{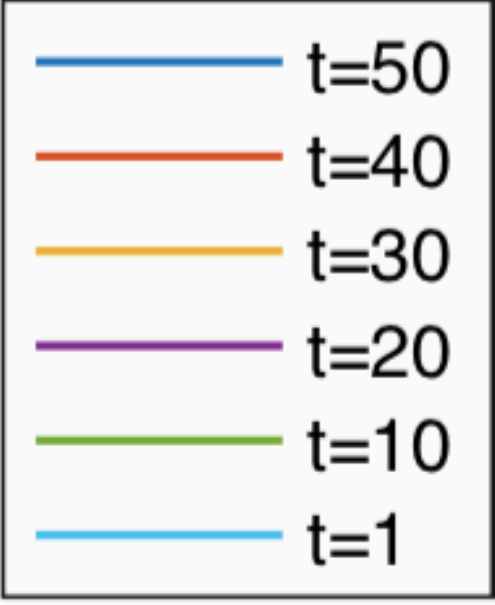}
\caption{Optimal firm behaviour (top panel: generation rate, bottom panel: trading rate) as a function of banked SRECs for various time-steps and SREC market prices. Parameters in Tables \ref{tbl:ComplianceParams} and \ref{tbl:ModelParams}.}
\label{fig:opt_beh}
\end{figure}

The most notable feature is the distinct regimes of generation/trading. For low levels of banked SRECs and near the terminal date, the firm generates/purchases until the marginal cost of producing/purchasing another SREC exceeds $P$, as the firm is almost assured to fail to comply. This follows the classic microeconomic adage of conducting an activity until the marginal benefit from the activity equals the marginal cost. In this regime, the marginal benefit of an additional SREC to the firm is $P$, as each additional SREC lowers their non-compliance obligation by $P$.

As the banked amount increases, the firm reaches a point where the marginal benefit from an additional SREC decreases from $P$. This occurs as the probability of compliance becomes non-negligible, as additional SRECs in excess of $R$ provide smaller marginal benefit than $P$. This is a result of the sale price of an SREC being bounded above by $P$ and leads to a decrease in optimal generation and optimal trading. The firm adjusts its behaviour so that its marginal costs are in line with this marginal benefit. This eventually leads to the firm selling as opposed to purchasing SRECs, as the net proceeds from the sale exceed the marginal value of retaining those certificates.

This decrease continues until the firm no longer benefits from additional SRECs. That is, at a certain level of banked SRECs $b$, the marginal benefit of an additional SREC is zero. Specifically, having an additional SREC does not increase the firm's likelihood of compliance, nor can they sell the additional SREC to make a profit. Accordingly, at all time-steps except $t = 50$, we observe that optimal generation jumps downwards once a certain level of $b$ is achieved. This would not occur if price impacts were inactive ($\eta = \psi = 0$). This is consistent with the theoretical results in \Cref{optimality}, where we showed that the optimal generation is either (i) greater than or equal to $h_t$, or (ii) identically $0$ (see \eqref{eq:g_eqn}). Recall, the condition for the firm to choose to generate is $P \PP_t(b_T^{g,\Gamma} < R) \geq - \psi \EE_t[\int_t^T \Gamma_u du]$. Thus, the firm chooses to shut down after reaching a threshold level of $b$ (which depends on both $t$ and $S$). Intuitively, this threshold is point at which an additional SREC is worth less to the firm than the impact that additional generation would have on the firm through its effect on $S$. As generation lowers SREC price, a firm that has already complied and plans to sell off remaining SRECs is pushing the market against themselves by continuing to generate. As such, there is a point where it is instead optimal to shut down production entirely and sell. This does not occur at $t = 50$, as it is the last decision point, and thus the price impact of generation does not impact the firm in any way.

Trading is influenced by a change in SREC price, which is in accordance with our intuition and aligns with the theoretical results from \Cref{optimality}. As SREC prices increase, the regulated firm chooses to purchase less, regardless of banked SRECs. We also see that higher SREC prices generally imply higher generation, as the firm chooses to generate their own SRECs, either to avoid paying high prices for them in the market, or to sell in the market and capitalize on the high prices (which of these two factors is the larger contributor depends on how much is banked).

If we hold $b, S$ constant, generation and purchasing are increasing in $t$. This is natural when the firm's compliance is not guaranteed, as with less time until the end of a compliance period, the firm needs to accumulate more SRECs in order to comply. For values of $b$ and $S$ for which compliance is guaranteed, we note that this property will not always hold, and is dependent on the value of $\gamma$. This is covered in more detail in \Cref{sens_gen_trade}. 

We also note that the optimal trading rate each of the `plateaus' varies slightly with time-step. This can be seen more clearly in \Cref{fig:sg_pi_demonstration},
where we enlarge the bottom-right plot in \Cref{fig:opt_beh} for two different areas of $b$ as an illustration of this property. That is, we plot optimal trading behaviour for each of the six time-steps detailed in \Cref{fig:opt_beh} when $S = 251$, $b \in [0, 40]$ (left) and $b \in [500, 1000]$ (right). 
\begin{figure}[!t]
\centering
\includegraphics[align = c, width=0.4\textwidth]{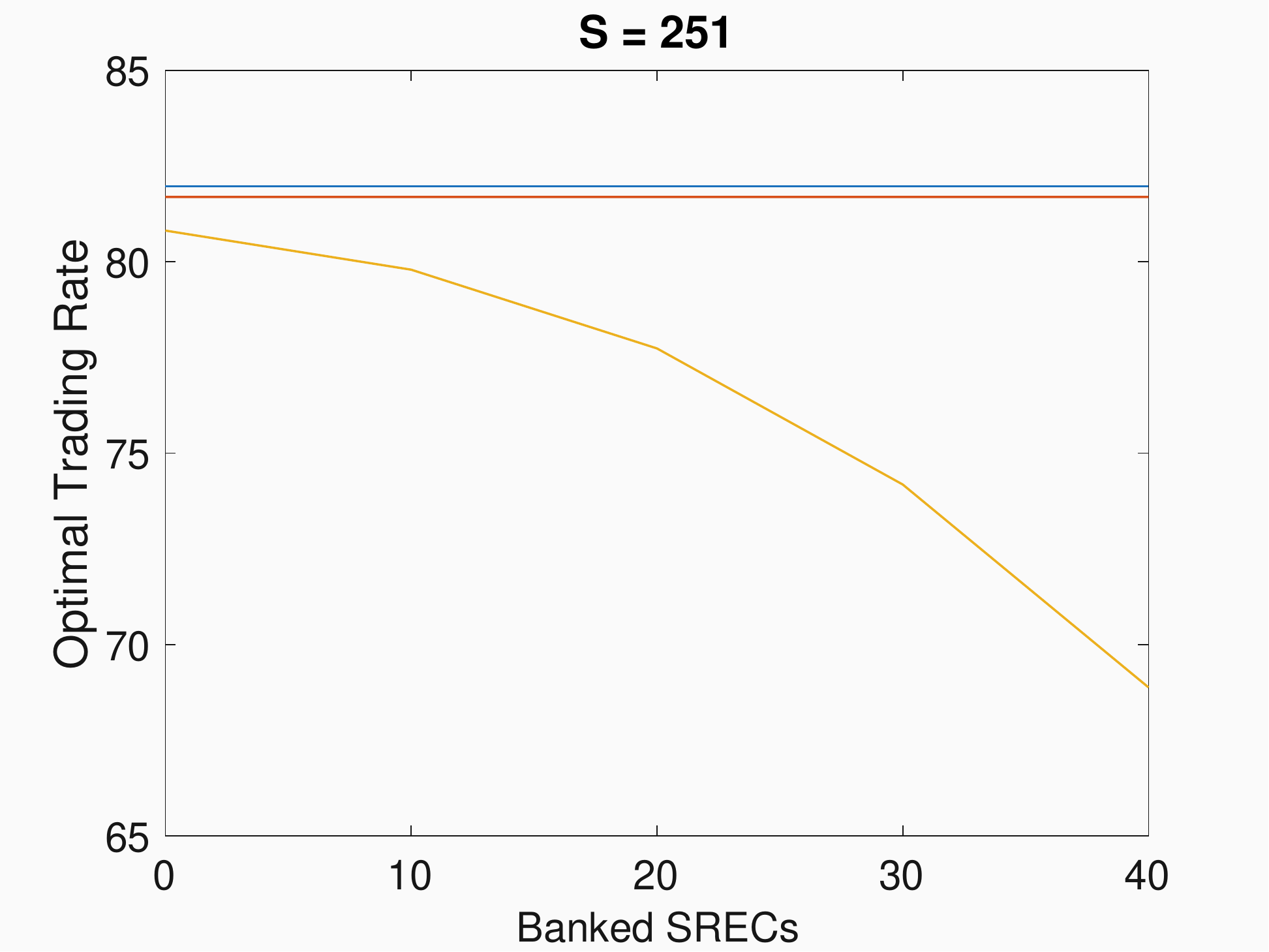}
\includegraphics[align = c, width=0.4\textwidth]{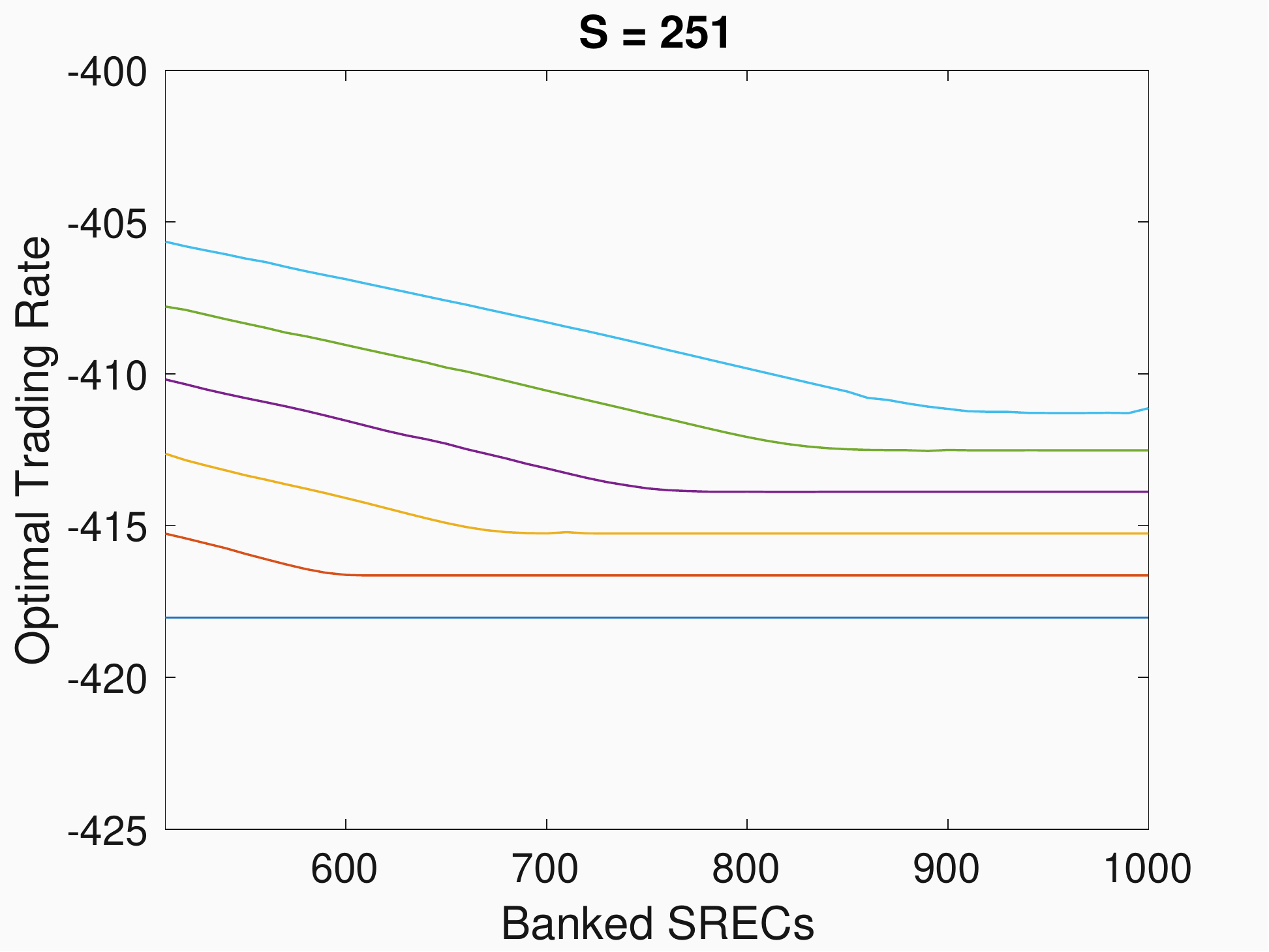}
\includegraphics[align = c, width=0.1\textwidth]{legend.pdf}
\caption{Optimal trading behaviour with price impact parameters $\eta=0.01$, $\psi=0.01$ for $b \in [0, 40]$ (left), $b \in [500, 1000]$ (right) and $S_0 = 250$ as a function of banked SRECs for various time-steps. All remaining parameters in Tables \ref{tbl:ComplianceParams} and \ref{tbl:ModelParams}.}
\label{fig:sg_pi_demonstration}
\end{figure}
As \Cref{fig:sg_pi_demonstration} illustrates, at high levels of banking ($b > R$), firms sell less at earlier time-steps than they do at later time-steps. Firms do this to mitigate the impact that their selling has on the SREC price and limit the extent to which the market moves against the firm as a result of their trading behaviour. The inverse behaviour occurs for low banking levels. Firms purchase less at earlier time-steps in order to keep prices down (relative to what would occur if they did not) and make compliance more attainable. At low banking levels, they also generate more, for the same reason (not shown to avoid repetition). These effects are proportional to the magnitude of $\eta$ and $\psi$, and increase with $S$. The generation rate for large values of $b$ does not vary with time-step due to the lower bound of generation being $0$. 

\subsection{Sample Paths} \label{sample_paths}

In \Cref{fig:path}, we show the dynamics of optimal firm behaviour through the compliance period. Here, $S_0 = 150$, $b_0 = 0$, and we simulate a path for $S$ and $b$ and at each time-step along this path, we adopt the optimal firm strategy in accordance with their banked SRECs and the SREC price.
\begin{figure}[t!]
\centering
\includegraphics[width=.75\textwidth]{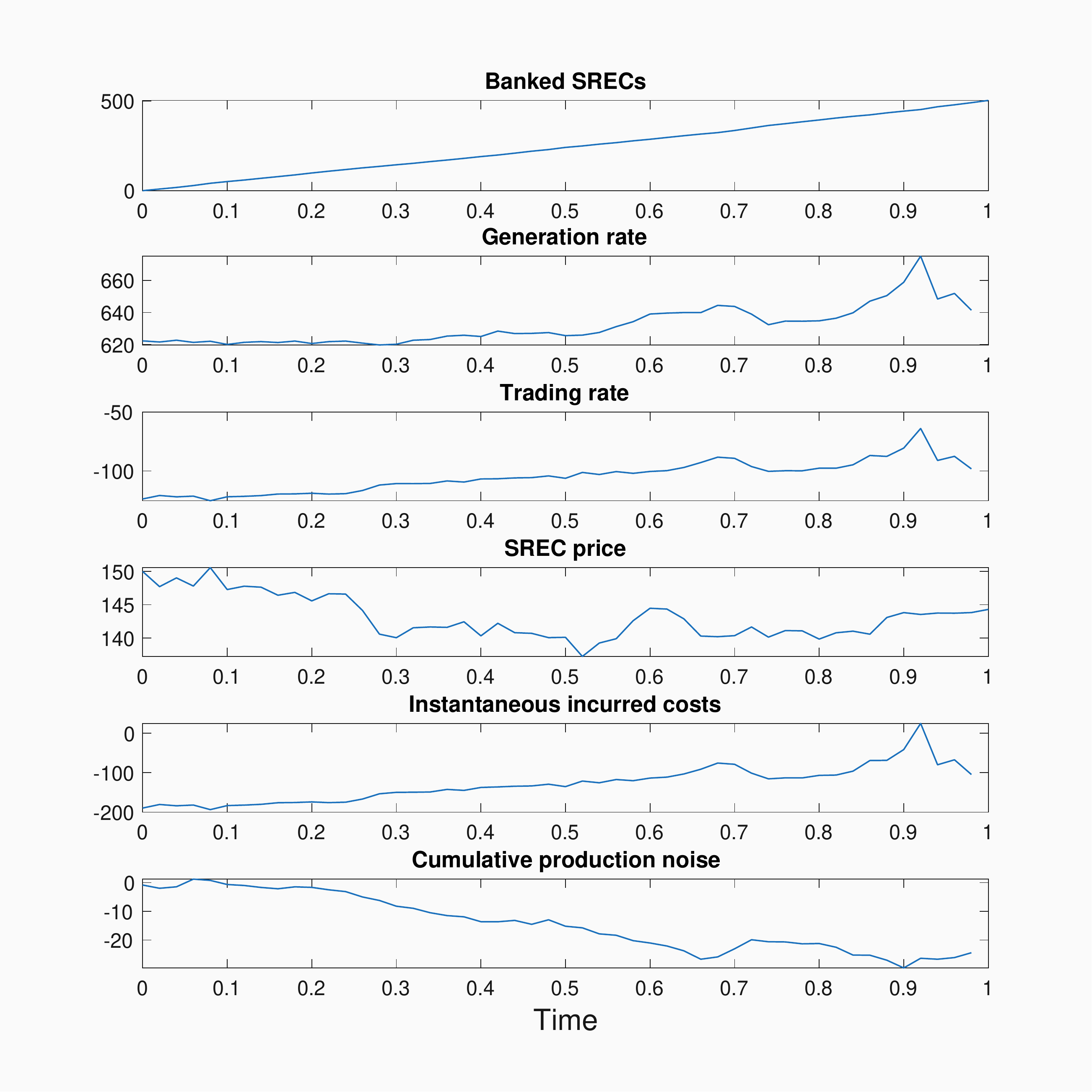}
\caption{A sample path of optimal firm behaviour with initial condition $S_0 = 150$ and $b_0 = 0$, and all remaining parameters in Tables \ref{tbl:ComplianceParams} and \ref{tbl:ModelParams}.}
\label{fig:path}
\end{figure}
From \Cref{fig:path}, the regulated firm banks SRECs at what appears to be a steady rate, and in this sample path, the firm reaches compliance. We will see shortly that the latter does not always occur. While banked SRECs appear to be linear, there is some variation in the amount of SRECs the firm banks at each time-step, which is the result of the SREC production noise the firm experiences. If we were to plot $b_t - \tfrac{Rt}{T}$ as a measure of the firm's SREC inventory versus the pro-rated amount they would need to be on-track to comply, we would see a roughly similar shape to the firm's cumulative production noise over the course of the period.


Turning our attention to the other subplots, we see the generation and trading processes exhibit notable variation over time. In particular, the inverse relationship between SREC price and trading rate is evident at earlier points in the period. Similarly, we can observe a positive relationship between SREC price and planned generation rate during the same time frame. However, as the period progresses,  generation and trading begin to move in the same direction, regardless of $S$ and its movements. This occurs as the randomness associated with SREC generation buffets the firm and changes their banked SRECs from one time-step to another in a way that cannot be foreseen. As $t$ approaches $T$, the firm has less time to adjust for this unforeseen noise resulting in the observed firm behaviour later in the period. The firm may have significantly more or less SRECs than what their planned generation and trading activity would suggest, and thus, they must determine whether they need to increase their SREC acquisition rate (increase planned generation and purchase more) or decrease their SREC acquisition rate (decrease planned generation and sell more) in order to behave optimally. We re-state that excess SRECs above $R$ expire valueless, so there is incentive for the firm to liquidate excess SRECs if in a strong position for compliance. 

\textcolor{black}{The SREC price itself is also pushed downwards throughout the period by the actions of the agent. As the agent is generating SRECs and selling them, the SREC price is lower than what it would be if we had set $\eta = \psi = 0$.} 

In Figure \ref{fig:path}, we see that cumulative production noise (the lowest subplot) decreases for the vast majority of the period. This means the firm generates less than planned in this time. As a reaction to this, they increase their planned generation and trading over the period in order to reach compliance, constantly reacting to their under-generation to put themselves back on track to achieve compliance. In general, towards the end of the period, increases (decreases) in cumulative production noise incite the firm to decrease (increase) planned generation and decrease (increase) purchasing of SRECs. 


The fifth panel in \Cref{fig:path} shows instantaneous incurred costs (IIC), which is the running cost incurred to the firm at each time-step:
\begin{equation}
    IIC_i = \left(\tfrac{\zeta}{2}  ((g_{t_i} - h_{t_i})_+)^2
+ \Gamma_{t_i} S_{t_i}^{g, \Gamma}
+ \tfrac{\gamma}{2} \Gamma_{t_i}^2\right) \Delta t.
\end{equation}
For the parameters chosen in Figure \ref{fig:path} and the resulting optimal behaviours, IIC is negative at all time-steps, which signifies that the firm is making a profit in the system, due to their sale of SRECs.

Next, by performing multiple simulations, we investigate the distribution of various quantities of interest, including total SRECs $b_T$, total planned generation $\int_0^T g_u du$, total traded amount $\int_0^T \Gamma_u du$, and total profit (negative of costs). For the base-line parameter choice, and with initial condition $b_0=0$, $S_0 = 150$, we present summary statistics using  $1,000$ simulated paths of $S$ and $b$ in \Cref{tbl:sum_stat_simple}.
\begin{table} [h]
\begin{center}
 \begin{tabular}{crrrrrr}
 \toprule\toprule
 Statistic & Mean & Std.Dev &1st Quartile  & 3rd Quartile & Skewness & Kurtosis \\
  \midrule
 $b_T$ & 501.61 & 1.62 & 500.50 & 502.74 & -0.02 & 2.69 \\
 $\int_0^T g_u du$ &621.95  & 6.59 & 617.46 & 626.25 & -0.001 & 3.02 \\
 $\int_0^T \Gamma_u du$ &-120.10  &6.31  & -124.34 & -115.79 & 0.02 & 2.86 \\
 Profit & 8,730.00  & 940.00 &  8,080.00 & 9,360.00 & 0.02 & 2.91 \\
\bottomrule\bottomrule
\end{tabular}
\end{center}
\caption{Summary statistics using 1,000 sample paths of $S$ of banked amount, generation, trading, and profit following the optimal strategy with initial condition $S_0=150$, $b_0=0$ and all remaining parameters in Tables \ref{tbl:ComplianceParams} and \ref{tbl:ModelParams}.} \label{tbl:sum_stat_simple}
\end{table}

In this one-period setup, the firm's optimal behaviour results in a symmetric distribution centred just above the requirement of $500$. There are cases (approximately 25\% of simulations) where the firm fails to comply ($b_T < 500$). Since $b$ (conditional on the firm's controls) is stochastic, and there is no advantage to additional SRECs above the requirement in a single-period framework, firms must strike a balance between being certain of compliance and wasting funds planning to generate or purchase SRECs over the requirement that may potentially end up unused. As such, for these parameters, the optimal firm plans to acquire (represented by $\int_0^T (g_u + \Gamma_u) du$) slightly more than the requirement of $500$, providing themselves with some buffer throughout the period in the event that they produce less than planned. However, this buffer is not so large that the firm is guaranteed to always comply.

\subsection{Parameter Sensitivity} \label{sensitivity}

In this section we investigate how varying parameters affect the optimal behaviour and resulting summary statistics, and explore the intuition behind the resulting effects.

\subsubsection{Sensitivity to Price Impact} \label{pimpacts}

In this section, we explore the impact of changing $\eta$ and $\psi$ on the optimal controls of the regulated firm. To do this, we compare an optimally behaving firm in a single-period model that is subject to various price impact scenarios to the baseline scenario of $\eta = \psi = 0.01$. We consider the $(\eta, \psi)$ pairs of $\{(0, 0), (0.01, 0.01), (0.02, 0.02)\}$. To do this, we simulate $1,000$ paths of $S$ in each price impact scenario, using the same random numbers in each scenario for $S$\footnote{While the random numbers used to generate paths are identical, the presence of price impact leads to different paths as impact varies. } and $\varepsilon$. In each path of $S$, we calculate total generation, total trading, and profit for the firm, and the difference between each quantity and their analogous amount under the baseline scenario. We calculate the mean and standard deviation of these differences across all paths, for each scenario. For example, for a pair $(\eta,\psi)$ we compute Profit($\eta,\psi$)-Profit($0.01,0.01$) across all scenarios and report the mean and standard deviation. In the first row of Table \ref{tab:p_impact_tbl} we report the raw results for  the case $\eta=\psi=0.01$, while rows 2--3 report the results for the difference relative to the benchmark  for the $2$ remaining pairs of $(\eta,\psi)$.

\begin{table}[htbp]
\color{black}
  \centering
    \begin{tabular}{rrrrrrrrrrr}
    \toprule
    \toprule
    \multicolumn{1}{c}{$\eta$} & \multicolumn{1}{c}{$\psi$} &       & \multicolumn{2}{c}{$\int_0^T g_u \,du$} &       & \multicolumn{2}{c}{$\int_0^T \Gamma_u \,du$} &       & \multicolumn{2}{c}{Profit} \\
\cmidrule{4-5}\cmidrule{7-8}\cmidrule{10-11}          &       &       & \multicolumn{1}{l}{mean} & \multicolumn{1}{l}{std.dev.} &       & \multicolumn{1}{l}{mean} & \multicolumn{1}{l}{std.dev.} &       & \multicolumn{1}{l}{mean} & \multicolumn{1}{l}{std.dev.} \\
\cmidrule{1-2}\cmidrule{4-5}\cmidrule{7-8}\cmidrule{10-11}
0.01 & 0.01 & & 621.95 & 6.59 & & -120.10 & 6.31 &  & 8,730 & 954.51 \\\midrule
0     & 0     &       &    3.69  &        0.25  &       & -    4.07  &             0.24  &       &        440  &         58  \\
    0.02     & 0.02 &       & -4.04  &        0.25  &       &             -    3.92  &             0.27  &       & -430  &           57  \\
    \bottomrule
    \bottomrule
    \end{tabular}%
\caption{Mean and standard deviation of differences in quantities of interest between an optimally behaving firm under various price impact scenarios and an optimally behaving firm subject to the baseline scenario of $\eta = \psi = 0.01$. We use 1,000 sample paths of $S$, with initial condition $S_0 = 150$ and remaining parameters as in Tables \ref{tbl:ComplianceParams} and \ref{tbl:ModelParams}.  }
  \label{tab:p_impact_tbl}%
\end{table}%


\textcolor{black}{Decreasing $\eta$ and $\psi$ to 0, which removes the impact of the regulated firm in the market altogether, allows the firm to generate more and sell more without any fear of pushing the price downwards and the market against them. This results in a higher profit during the compliance period.}

Increasing $\eta$ and $\psi$ to 0.02 each results in lower generation, less selling, and consequently, lower profit. This is the result of the firm attempting to mitigate their price impact through sales, resulting in lower generation as a consequence (so as to not end up with a large amount of surplus SRECs). In general, price impacts lead to a feedback loop, as the firm's behaviour of generating and selling lowers prices, which further incentivize decreased selling and planned generation (as seen in \Cref{fig:opt_beh}).

\subsubsection{Sensitivity to Trading and Generation Costs}\label{sens_gen_trade}

To conclude our analysis of the single period model, we explore sensitivity to generation and trading speed costs ($\zeta$ and $\gamma$). \Cref{fig:TradingGenerationCost} shows how the optimal behaviour changes for various values of
$\zeta$ and $\gamma$, across six time-steps, for fixed SREC price level $S_t = 150$.

\begin{figure}[t!]
\centering
\begin{subfigure}{0.26\textwidth}
\includegraphics[align = c,width=\textwidth]{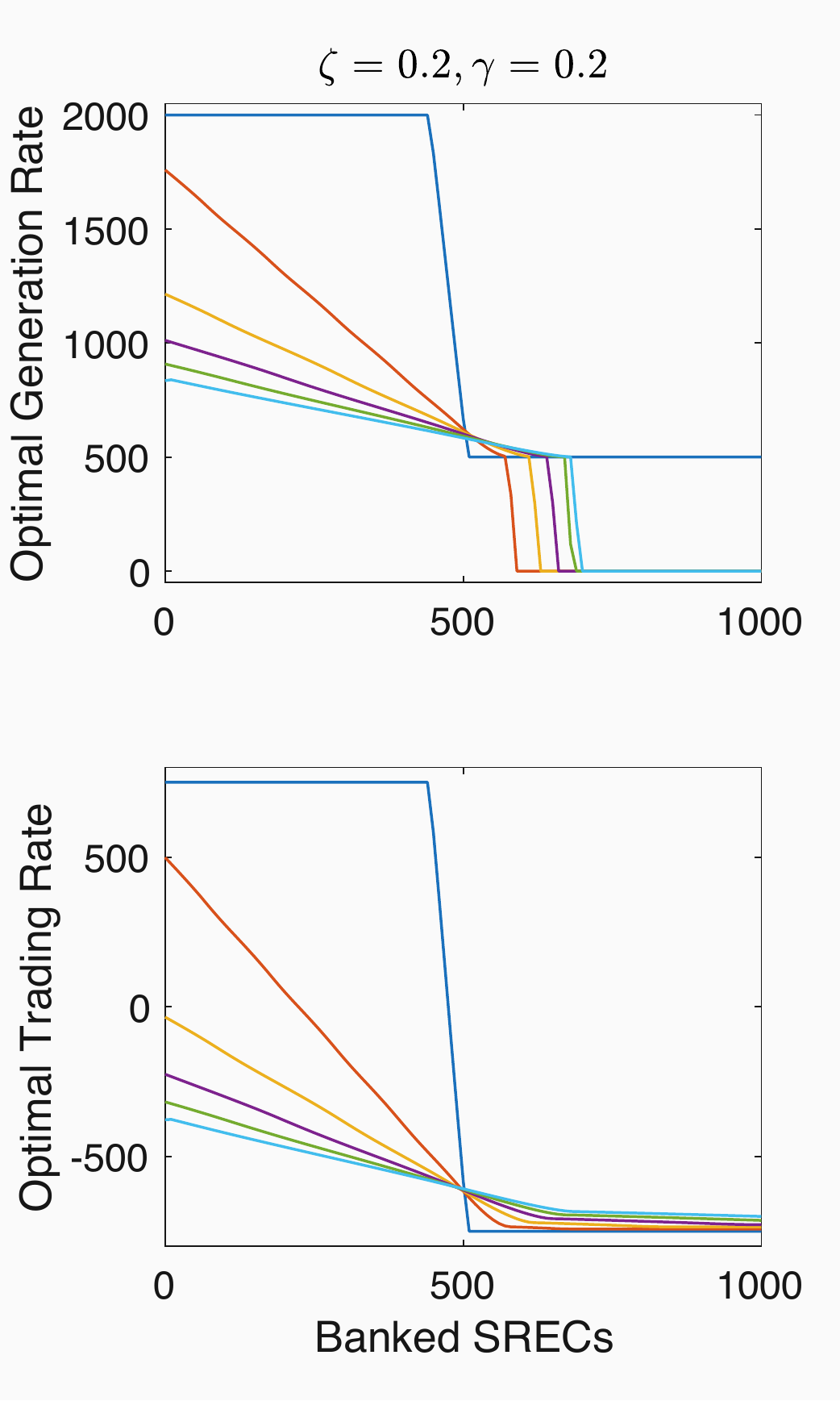}
\caption{\footnotesize$(\zeta, \gamma) = (0.2, 0.2)$}
\end{subfigure}
\begin{subfigure}{0.26\textwidth}
\includegraphics[align = c,width=\textwidth]{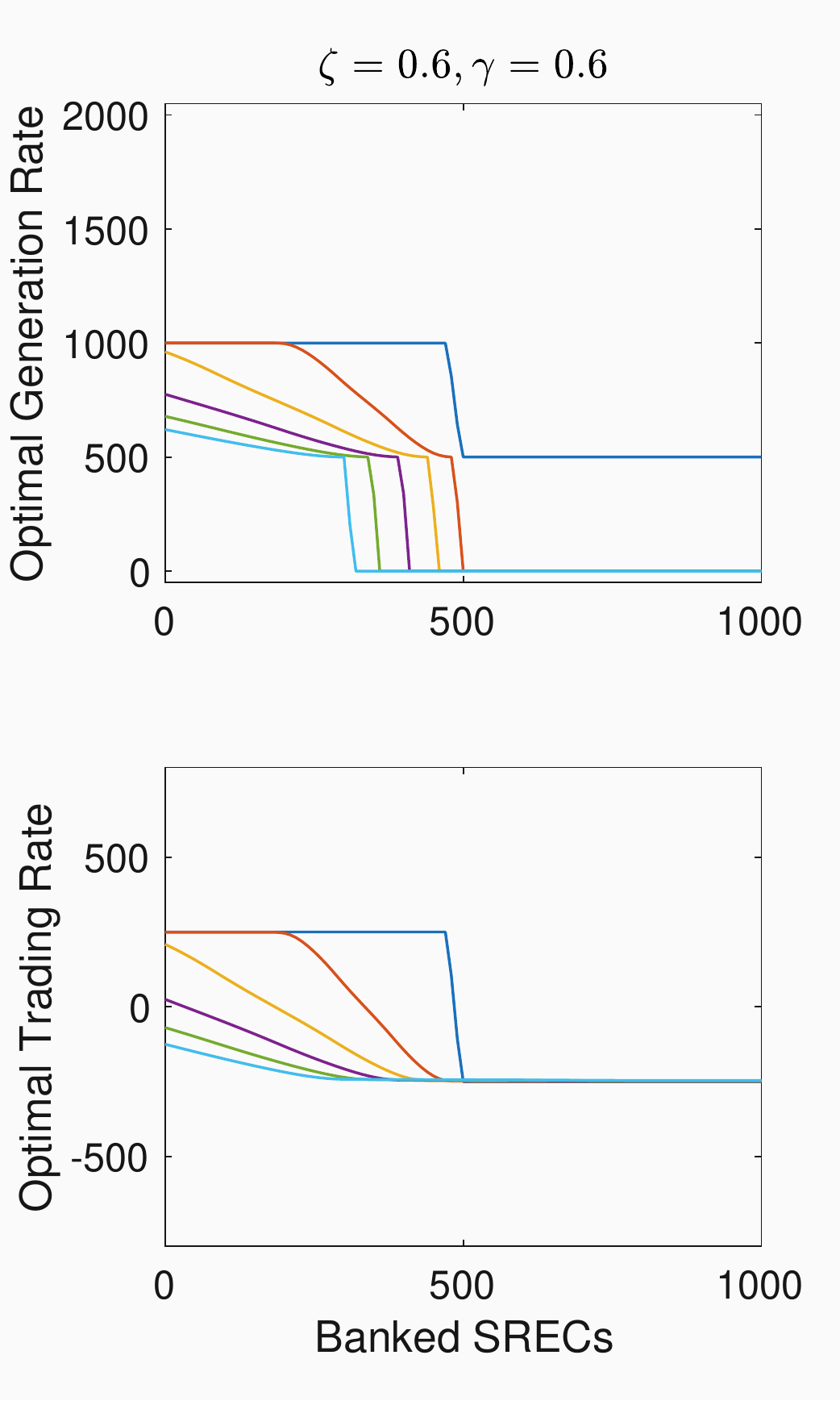}
\caption{\footnotesize$(\zeta, \gamma) = (0.6, 0.6)$}
\end{subfigure}
\begin{subfigure}{0.26\textwidth}
\includegraphics[align = c,width=\textwidth]{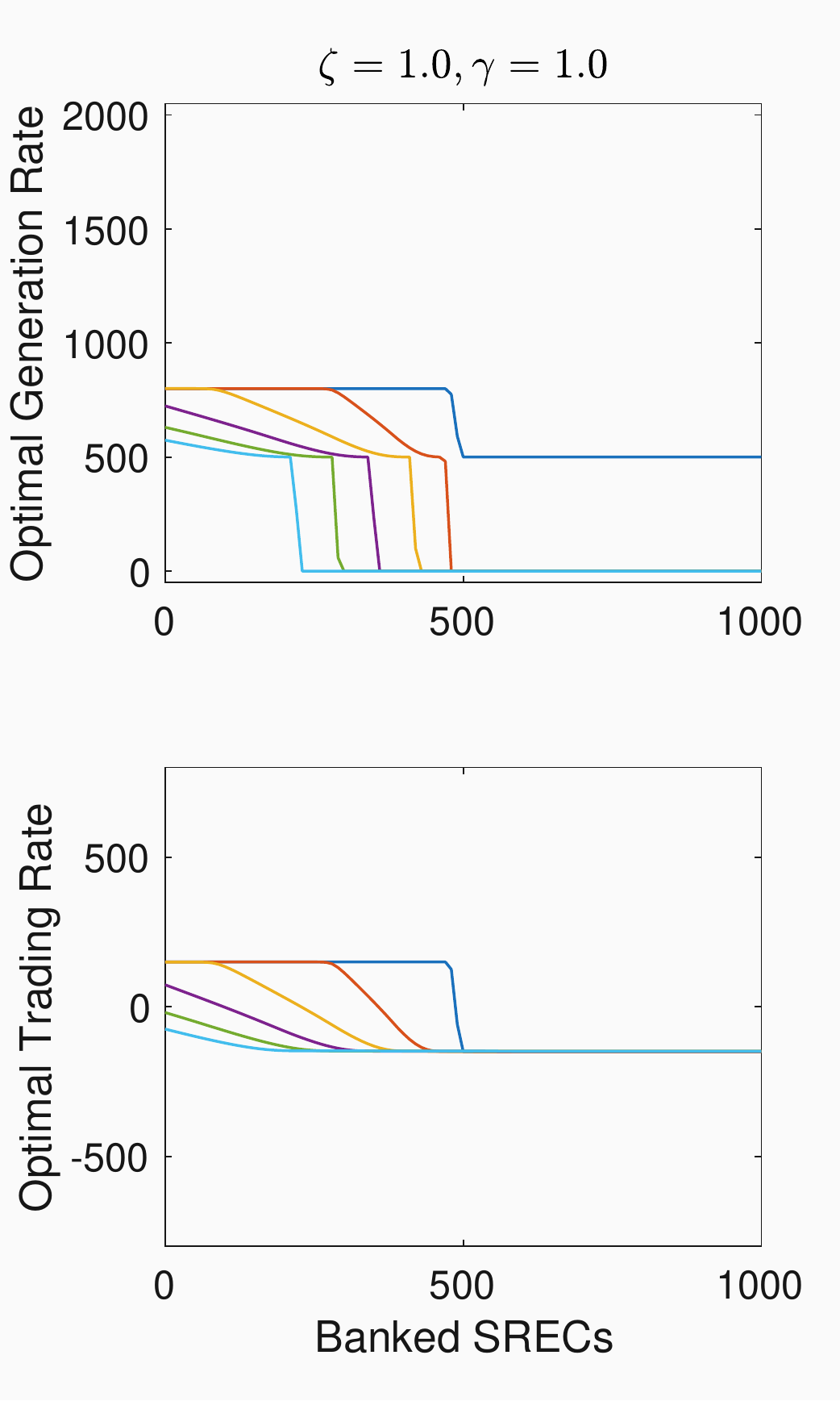}
\caption{\footnotesize$(\zeta, \gamma) = (1, 1)$}
\end{subfigure}
\begin{subfigure}{0.1\textwidth}
\includegraphics[align = c,width=\textwidth]{legend.pdf}
\end{subfigure}
\caption{Optimal generation and trading rates for differing levels of $\zeta$ (generation cost parameter) and $\gamma$ (trading speed penalty parameter) when $S_t=150$.   Remaining parameters as in Tables \ref{tbl:ComplianceParams} and \ref{tbl:ModelParams}.}
\label{fig:TradingGenerationCost}
\end{figure}

The middle subplots in \Cref{fig:TradingGenerationCost}  show the firm's optimal behaviour in the default setting of $\zeta = 0.6, \gamma = 0.6$ as in \Cref{tbl:ModelParams}. Increasing/decreasing  $\zeta, \gamma$ compresses/expands the range of optimal trading and planned generation.  This is the result of higher/lower parameters corresponding to higher/lower costs and decreased/increased capacity of the firm to invest in generation and to trade.


Finally,  \Cref{fig:TradingGenerationCost}(a) shows that when $b$ is above $R$, optimal planned generation and purchasing are larger at earlier time-steps than later time-steps. This is the result of small $\gamma$ leading to low trading costs, and the firm can  aggressively sell excess SRECs  before $T$. Hence, at earlier time-steps, the firm continues to generate above their baseline in order to acquire more SRECs to sell later in the period. Later in the period, the firm prefers to liquidate their excess SRECs in order to ensure they do not have excess inventories at time $T$, resulting in the observed behaviour. This does not happen in the cases where $\gamma = 0.6$ or $1$ as the firm is limited in how quickly it can viably liquidate excess SRECs by its trading speed penalty.

We do not include the plots of $(\zeta, \gamma)$ combinations where $\zeta \neq \gamma$ to avoid repetition. The results and interpretation are identical to those discussed above, with changes in $\zeta$ impacting optimal generation and changes in $\gamma$ impacting optimal trading.

\subsection{Multi-period model} \label{multi_per}

Thus far, we have considered a single period compliance framework. In practice, SREC markets consist of multiple periods. In this section, we present the results for an $N$-period SREC market, which is described in \Cref{model}. Much of the behaviour and intuition discussed in the earlier parts of this section carry over to the multi-period case. For the multi-period formulation, we assume  there are $n$ (equally spaced) decision  points within each compliance period denoted
\begin{equation}
    0 = t_1 < \dots < t_n < T_1 = t_{n+1} < \dots < t_{2n} < T_2  = t_{2n +1} < \dots < t_{nN} < T_N = t_{nN +1},
\end{equation}
where $t_k = k\Delta t$. The last time-step $t_{Nn+1}$ is not a decision point. Therefore, there are $n\times N$ decision points, from $t_1, ..., t_{Nn}$. We will use the notation $\mathcal{T}:=\{T_1,\dots,T_N\}$ to denote the set of compliance times. 

As before, we continue assuming $P$ and $R$ are constant across each of the $N$ periods, and  the processes $g_t, \Gamma_t$ are piecewise constant within $[t_i, t_{i+1})$, with the firm controlling $\{g_{t_i}, \Gamma_{t_i}\}_{i \in \mathfrak{R}}$, where $\mathfrak{R} = \{0, ..., n\times N\}$. As in \Cref{discrete},  regulated firm choose their trading and generating behaviour at the start of the time interval.

The end points of the $i$-th period is $T_i$, $i=1,\dots,N$, and  firms may bank unused certificates with no expiry. In real SREC markets, certificates generally have a finite life-time, but allowing indefinite banking reduces the dimensionality of the problem significantly and renders it computationally tractable.  
The performance criterion (corresponding to the total cost) for an arbitrary admissible control is
\begin{align} \label{eq:discretepc_multi}
\begin{split}
J^{g, \Gamma}(k, b, S)=&
\EE_{t_k,b,S}\biggl[\sum_{i=k}^{Nn} \left\{ \tfrac{\zeta}{2} ((g_{t_i} - h_{t_i})_+)^2   + \Gamma_{t_i} S_{t_i}^{g, \Gamma} + \tfrac{\gamma}{2} \Gamma_{t_i}^2\right\} \Delta t  \\ 
&\quad\quad+ \sum_{j=1}^N P (R - b_{t_{nj}}^{g, \Gamma} - \Delta t (g_{t_{nj}} + \Gamma_{t_{nj}})-\nu \sqrt{\Delta t}\, \varepsilon_{t_{nj+1}})_+ \, \Id_{\{t_k < t_{nj +1}\}}\biggr].
\end{split}
\end{align}
The dynamics of the state variables ($b, S$) are modified as follows
\begin{subequations}
\begin{align}
S_{t_i}^{g, \Gamma} &= \min\left(\left(S_{t_{i-1}}^{g, \Gamma} + \left(\mu  + \eta \,\Gamma_{t_{i-1}}  - \psi\, g_{t_{i-1}}\right) \Delta t - \psi \nu \sqrt{\Delta t}\, \varepsilon_{t_i} + \sigma \sqrt{\Delta t}\, Z_{t_i}\right)_+\;,\; P\right) \label{eq:S_discrete_multi}
\\
b_{t_i}^{g, \Gamma} &= 
\left\{
\begin{array}{ll}
b_{t_{i-1}}^{g, \Gamma} +  (g_{t_{i-1}} + \Gamma_{t_{i-1}})\Delta t + \nu \sqrt{\Delta t}\, \varepsilon_{t_i}, & t_i\notin \mathcal{T}
\\[0.5em]
\left( b_{t_{i-1}}^{g, \Gamma} +  (g_{t_{i-1}} + \Gamma_{t_{i-1}})\Delta t + \nu \sqrt{\Delta t}\, \varepsilon_{t_i}-R\right)_+, & t_i\in \mathcal{T},
\end{array}
\right.
 \label{eq:b_discrete_multi} 
\end{align}%
\label{eqn:DiscreteStateEvolutionMulti}
\end{subequations}%
where $Z_{t_i}, \varepsilon_{t_i} \sim N(0, 1)$, iid, for all $i \in \mathfrak{N}$.

As in the single-period case, we seek
\begin{align}
V(t, b, S) = \inf_{g, \Gamma\in\mcA} J^{g, \Gamma}(t, b, S), \label{eq:DiscreteValueFuncMulti}
\end{align}
and the strategy that attains the inf, if it exists. Applying the Bellman Principle to \eqref{eq:DiscreteValueFuncMulti} implies
\begin{subequations}
\begin{align}
\begin{split}
V(t_i, b, S) &= \inf_{g_{t_i}, \Gamma_{t_i}} \biggl\{
\left(\tfrac{\zeta}{2}  ((g_{t_i} - h_{t_i})_+)^2
+ \Gamma_{t_i} S_{t_i}^{g, \Gamma}
+ \tfrac{\gamma}{2} \Gamma_{t_i}^2\right) \Delta t
\\ & \qquad\qquad
+ \EE_{t_i}\left[ P (R - b_{t_{i}}^{g, \Gamma} - \Delta t (g_{t_i} + \Gamma_{t_i}) - \nu \sqrt{\Delta t} \epsilon_{t_{i+1}})_+ \right]  \Id_{\{t_{i+1}\in\mathcal{T}\}}
\\ & \qquad\qquad
+ \EE_{t_i}\left[V\left(t_{i+1}, b_{t_{i+1}}^{g, \Gamma}, S_{t_{i+1}}^{g, \Gamma}\right)\right] \biggr\}, \qquad\qquad \text{and}
\label{eq:Discrete-Bellman_multi}%
\end{split}%
\\
V(T_N, b, S) &= 0. \label{eq:Discrete-Bellman-terminal_multi}
\end{align}%
\label{eqn:DiscreteBellman_multi}%
\end{subequations}%

The dynamics of $b$ in the multi-period framework are such that $b_{T_j}$ represents the firm's SRECs \textbf{after} submitting the compliance requirement for the compliance period ending at $T_j$. We adjust our solution algorithm described in \Cref{num_scheme} to account for the assumptions stated above, using the same model parameters, and choosing $N = 5$. We denote the current period by $m$. As the algorithm for obtaining the optimal controls in the multi-period problem is very similar to that detailed in \Cref{num_scheme}, we omit it here.

\subsubsection{Sample results in the Multi-period model} \label{mpm_pi}

Analogous to \Cref{fig:opt_beh}, \Cref{fig:multi-period-impact} 
shows the optimal behaviour of a regulated firm as a function of banked SRECs, across three different prices of $S$ and at six points in time during the first compliance period when there is price impact.
\begin{figure}[h]
\centering
\includegraphics[align = c, width=0.6\textwidth]{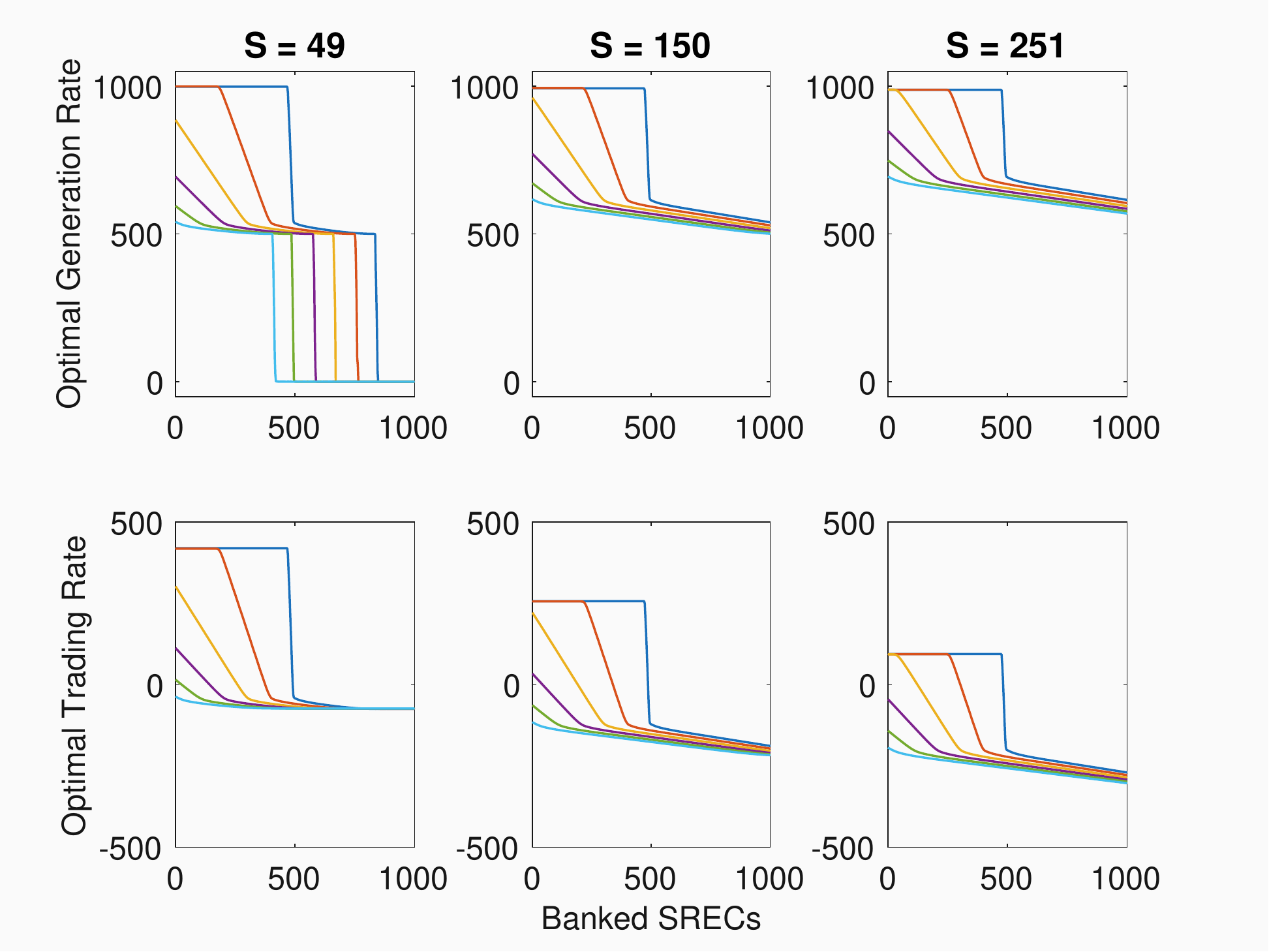}
\includegraphics[align = c, width=.1\textwidth]{legend.pdf}
\caption{Optimal firm behaviour as a function of banked SRECs across various time-steps (during the first of five compliance periods) and SREC market prices with parameters as in Tables \ref{tbl:ComplianceParams} and \ref{tbl:ModelParams}.}
\label{fig:multi-period-impact}
\end{figure}

In \Cref{fig:multi-period-impact}, we plot the dependence of the optimal generation and trading rate of the firm in the first period ($m = 1$) of the $5$-period model against banked SRECs, for three SREC prices, at six points in time, with all remaining parameters as in Tables \ref{tbl:ComplianceParams} and \ref{tbl:ModelParams}.
Much of the intuition surrounding \Cref{fig:opt_beh} applies here. There are, however, obvious differences between Figures \ref{fig:opt_beh} and \ref{fig:multi-period-impact}. As before, for low levels of banked SRECs, across all values of $S$, and near the end of the compliance period, the firm generates until the marginal cost of producing another SREC exceeds $P$, and purchases until the marginal cost of purchasing another SREC exceeds $P$, as the firm is almost assured to fail to comply. In this regime, the marginal benefit of an additional SREC   is $P$, as each additional SREC lowers their non-compliance obligation by $P$.

As the banked amount increases, the firm reaches a point where the marginal benefit from an additional SREC decreases from $P$. This occurs as the probability of compliance becomes non-negligible, as additional SRECs in excess of $R$ provide smaller marginal benefit than $P$. This leads to a decrease in optimal generation and optimal trading, as the firm adjusts its behaviour so that its marginal costs are in line with this marginal benefit. Thus far, this is the same interpretation as the single-period setting.
As $b$ continues to increase, the firm holds sufficient banked SRECs such that they will be able to acquire surplus certificates above $R$. These surplus SRECs have little value in the current period to the firm, even including their use as insurance for extreme under-generation. They may, however, bank SRECs putting the firm in a better position for future compliance periods. In the single-period case, at the end of the compliance period, holding additional SRECs lack utility. The concept of banking means that this is not true in the multi-period case, and thus we see an abrupt change in the slope of the optimal controls, and a slower decay in generation and purchasing rate when compared to \Cref{fig:opt_beh}.

This decrease continues until the firm no longer benefits from additional SRECs. That is, at a certain level of $b$, the marginal benefit of an additional SREC is zero. Specifically, having an additional SREC does not increase the firm's likelihood of compliance in current or future periods, nor can the firm sell the additional SREC for a profit (taking into account their trading costs and $S$). As in \Cref{fig:opt_beh}, this results in optimal generation dropping to $0$ and optimal trading plateauing at the level where the marginal revenue from trading equals the marginal cost. This plateau is not visible in every subplot in \Cref{fig:opt_beh_mult} due to axis limits and the fact that $m = 1$. The impact of SREC price on generation and trading is similar to the single period case.

As $m$ increases, the firm has fewer future periods to position themselves for. Consequently, the firm's optimal planned generation and purchasing behaviour decays more quickly for larger $m$. See Appendix \ref{additional_figures},  \Cref{fig:multi-period-23_pi,fig:multi-period-45_pi} for the analogous figures for $m = 2$, $3$, $4$, and $5$. The optimal controls when $m=5$ are identical to the single-period case as they must be since the performance criterion is time-consistent.

\Cref{fig:5_per_path} shows a sample path of the optimal strategy for three firms (with the same cost functions, and experiencing the same randomness in $b$ and $S$) throughout the course of the 5-period SREC market, with each period lasting $1$ year.  The firms differ in their initial banked amount: Firm 1 has $b_0 = 0$, Firm 2 has $b_0 = 250$ and Firm 3 has $b_0 = 500$. We set  $S_0= 150$. At each time-step, each firm behaves optimally given their values of banked SRECs and the SREC price. Each firm exists in a separate `universe' and they do not have an impact on one another. Moreover, they are subject to the same realized randomness from the Brownian motions impacting the SREC price and their own SREC generation.
\begin{figure}[!t]
\centering
\includegraphics[width=0.75\textwidth]{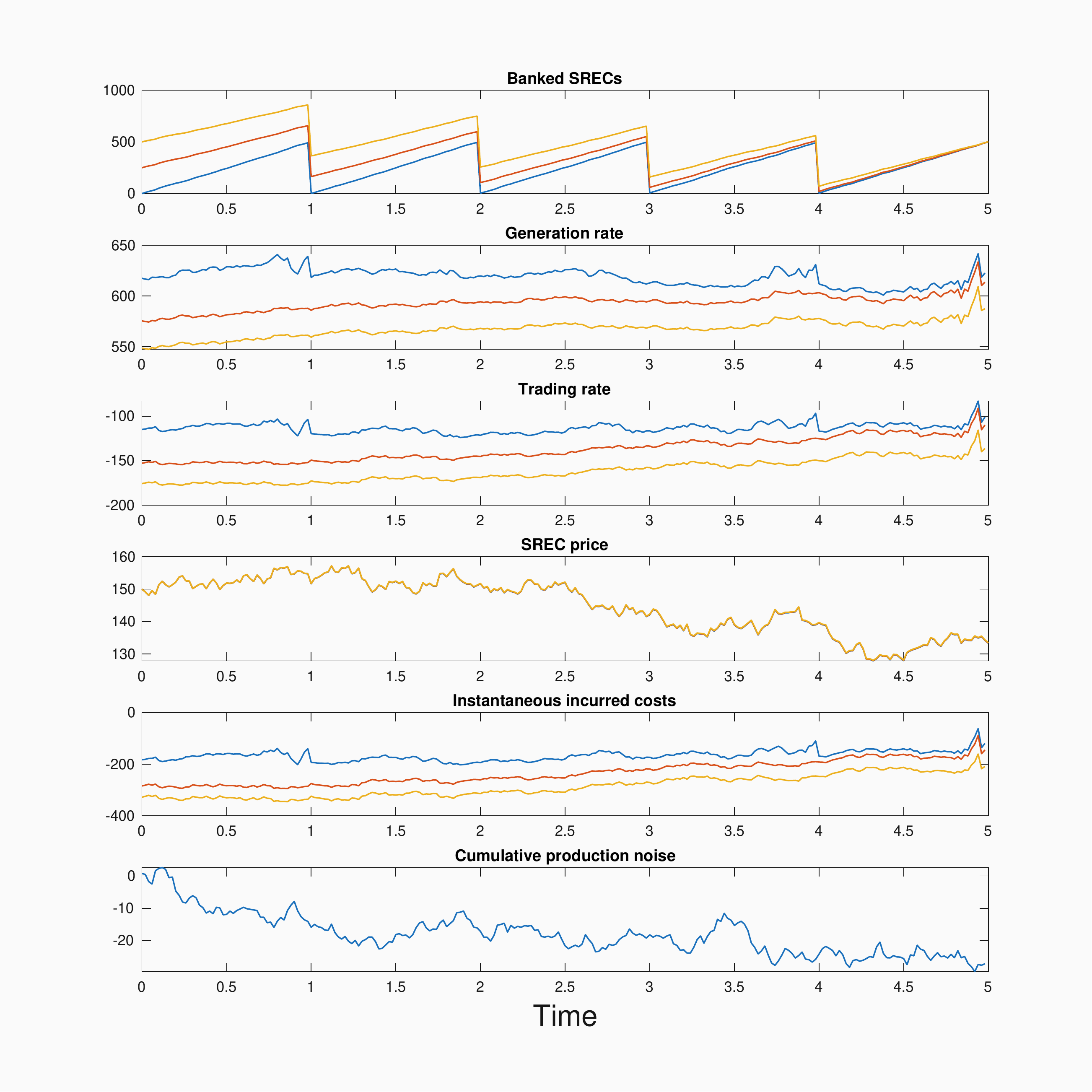}
\caption{Paths of three optimally behaving firms in a 5-period compliance system with $S_0 = 150, b_0 = 0$ (blue), $b_0 = 250$ (red), $b_0 = 500$ (yellow). Parameters as in Tables \ref{tbl:ComplianceParams} and \ref{tbl:ModelParams}.}
\label{fig:5_per_path}
\end{figure}

We see the banked SRECs for all three firms converge roughly to $R=500$ as $t\rightarrow 5$. Consequently, Firm 3 accumulates SRECs at a slower rate than Firm 2, who accumulates SRECs at a slower rate than Firm 1. Even with the firm impacted by production noise, the path of $b$ appears steady within each compliance period for the firms, as before. The large drops are the effect of the firm submitting SRECs for compliance at the end of each period. This results in the \textit{converging saw-tooth} pattern in the first subplot of \Cref{fig:5_per_path}.

\textcolor{black}{The optimal behaviours of each firm follow roughly the same pattern, suggesting that they react similarly to changes in $S$. The difference in their behaviours is primarily due to their initial banked SRECs $b_0$. Firm 1 has no spare SRECs at $t = 0$, and generates the most and sells the least. Firm 3 has $500$ spare SRECs at $t = 0$ -- enough for an entire period of compliance. As such, they produce the least and sell the most. Firm 2 operates between Firm 1 and Firm 3. Naturally, Firm 3 profits the most from this system, due to their initial position. All three firms slow down generation and purchasing behaviour near the final time-steps, reacting to unexpected generation noise that has resulted in them generating more than planned in the time-steps immediately prior. This occurs at the ends of non-terminal compliance periods for Firm 1, as they are typically right on the border of compliance at each period, due to their small initial inventory. The other firms have SREC balances above $R$ and, as  banking is allowed, there is no need for a firm to liquidate excess banked SRECs early. }

\textcolor{black}{The optimal behaviours of each firm also imply different SREC prices in each of the `universes' that each firm exists in. However, the magnitude of price impacts for an individual firm are small enough that visually, the price processes look almost identical. In fact, there is a difference of about \$0.07\footnote{At the terminal time-step} between the price path for Firms 1 and 2, and \$0.10 between the price path for Firms 2 and 3. Firm 3 has the highest price, as they are taking the least extreme generation and trading behaviour. Firm 1 has the lowest price, for the opposite reason. }

Finally, we simulate many paths of $S$ with $S_0 = 150, b_0 = 250$ in order to obtain summary statistics and learn about the distribution of various quantities for each firm. In Figure \ref{fig:histograms_mult_price_impact} we plot the histograms of total generated SRECs and total traded SRECs for a regulated firm in each period, based on $1,000$ such sample paths.

\begin{figure}[h!]
\centering
\includegraphics[width=0.7\textwidth]{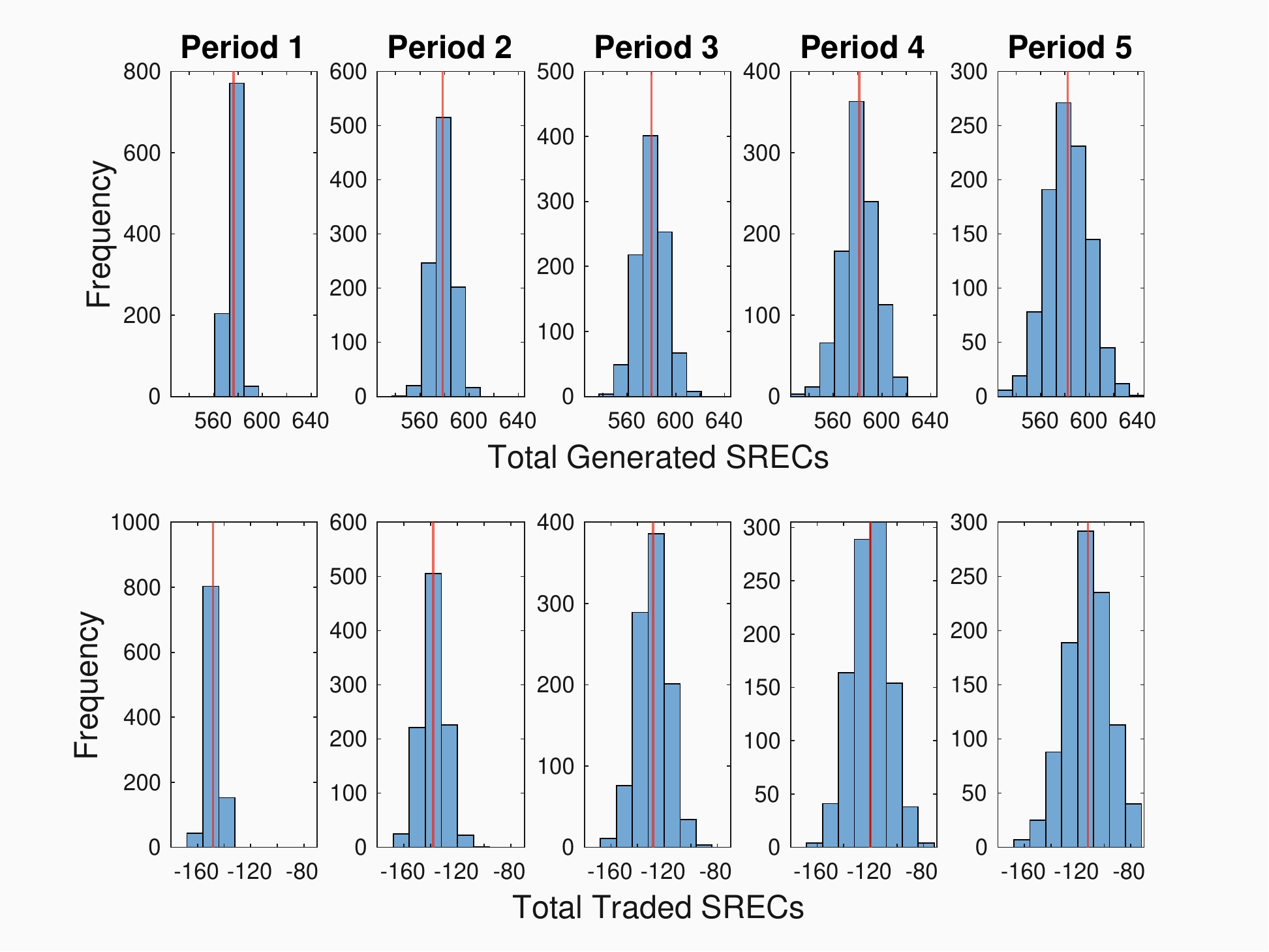}
\caption{Histogram of firm generation and trading across each compliance period with $S_0 = 150, b_0 = 250$. Parameters as in Tables \ref{tbl:ComplianceParams} and \ref{tbl:ModelParams}.}
\label{fig:histograms_mult_price_impact}
\end{figure}
From the figure, we note that aggregate selling decreases as $m$ increases, while total planned generation is relatively more static. In particular, the static nature of $\int_0^T g_u du$ arises because lower values of $m$ are associated with higher levels of excess SRECs, as the firm begins with $b_0 = 250$ and thus has the freedom to plan to generate slowly. The change in trading is the result of the firm reacting to the (generally) lower SREC prices that occur when price impacts are active. We also see that the variance of the firm's aggregate behaviour increases as the periods progress. This is the result of simulating forward paths of $S_t$ conditioning on $\mathcal{F}_0$, as $\text{Var}(S_t | S_0)$ is increasing in $t$. As before, these patterns persist across various choices of $S_0$ and $b_0$. To avoid repetition, plots for other initial conditions are not included in this work.
\section{Conclusion}

In this work, we characterize the optimal behaviour of a single regulated LSE in a single-period SREC market. In particular, we characterize their optimal generation and trading behaviour as the solution to a continuous time stochastic control problem. In doing so, we characterize the solution and tease out essential features of the optimal strategy. We also numerically solve for the system in a discrete time setting for both single and multi-period SREC frameworks. Through this, we provide intuition and reasoning for the resulting optimal behaviour, including detailed analysis of various sample paths, summary statistics, strategies, and parameter choices.

Many further extensions are possible. Interactions between agents are a critical component of real SREC markets that are largely ignored in this single-firm setup. In particular, incorporating partial information of firms would be a very challenging but mathematically interesting problem that would more closely mimic the realities of SREC markets. This could potentially necessitate the use of a mean field games approach. Improved calibration to real world parameters would also increase the applicability of this work for use by regulated firms and regulators. \textcolor{black}{The privacy of the relevant data needed to accurately calibrate the cost parameters presents a significant challenge to this endeavour.}

However, even our simple model reveals salient facts about the nature of these systems and how firms should behave when regulated  by them. Our single-period model reveals that the optimal generation and trading of regulated firms broadly exists in three regimes, depending on the marginal benefit received from holding an additional SREC. We observe that a firm's trading behaviour is more sensitive to changes in $S$ than its generation behaviour, and that higher SREC prices imply greater generation and lower purchasing (more selling). We show consistency between the numerical and theoretical solutions for our model. In particular, the interesting property that firms should generate above their baseline or shut down entirely is clearly demonstrated theoretically and empirically. Furthermore, we discuss sensitivity to selected other parameters in our model.

When extending to the multiple-period framework, we observe many similarities, but also the key difference that a fourth regime exists in the optimal generation and trading of regulated firms; that is, the regime where a marginal SREC does not provide value in the current period, but may be banked to provide value in the future. Additionally, we compare and contrast the optimal behaviours of firms throughout the multiple-period framework based on different initialization points, and study the changes in their aggregate behaviour across compliance periods.

In providing these results, we have produced a framework and numerical solution that would be of use for both regulated firms and regulatory bodies who both have immense interest in understanding the optimal behaviour of regulated LSEs in these systems.
\nocite{*}
\bibliographystyle{siamplain}
\bibliography{references}
\appendix 

\section{Additional Figures} \label{additional_figures}

Included below are plots of the regulated firm's optimal behaviour in the context of \Cref{mpm_pi}, for periods 2-5 of a 5-period model, with price impacts active. In all cases, the legend in Figure \ref{fig:opt_beh} applies.

\begin{figure}[h!]
\centering
{\large$\boldsymbol{m = 2, 3}$\qquad\qquad\quad}
\\
\includegraphics[align = c, width=0.4\textwidth]{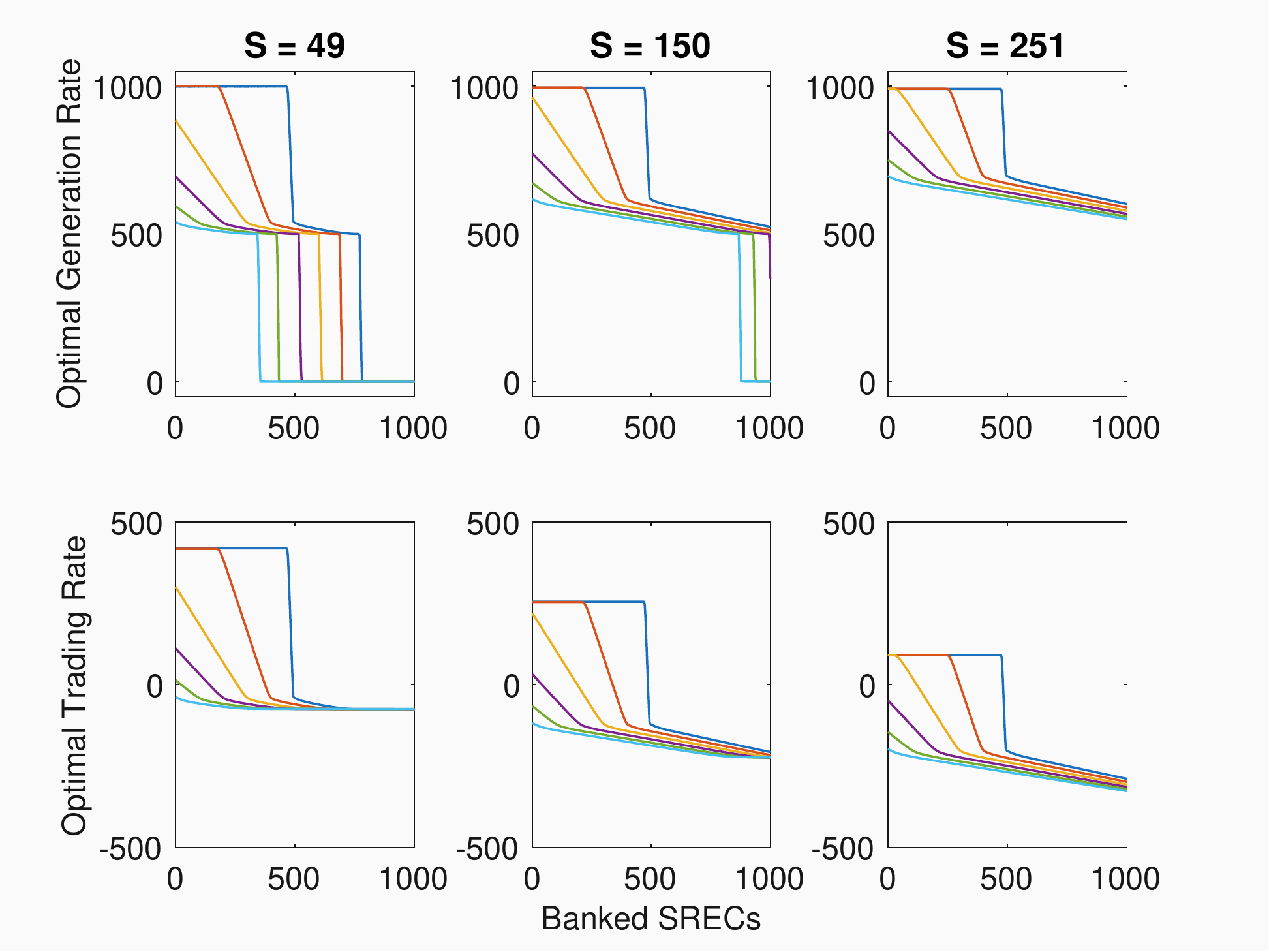}
\includegraphics[align = c, width=0.4\textwidth]{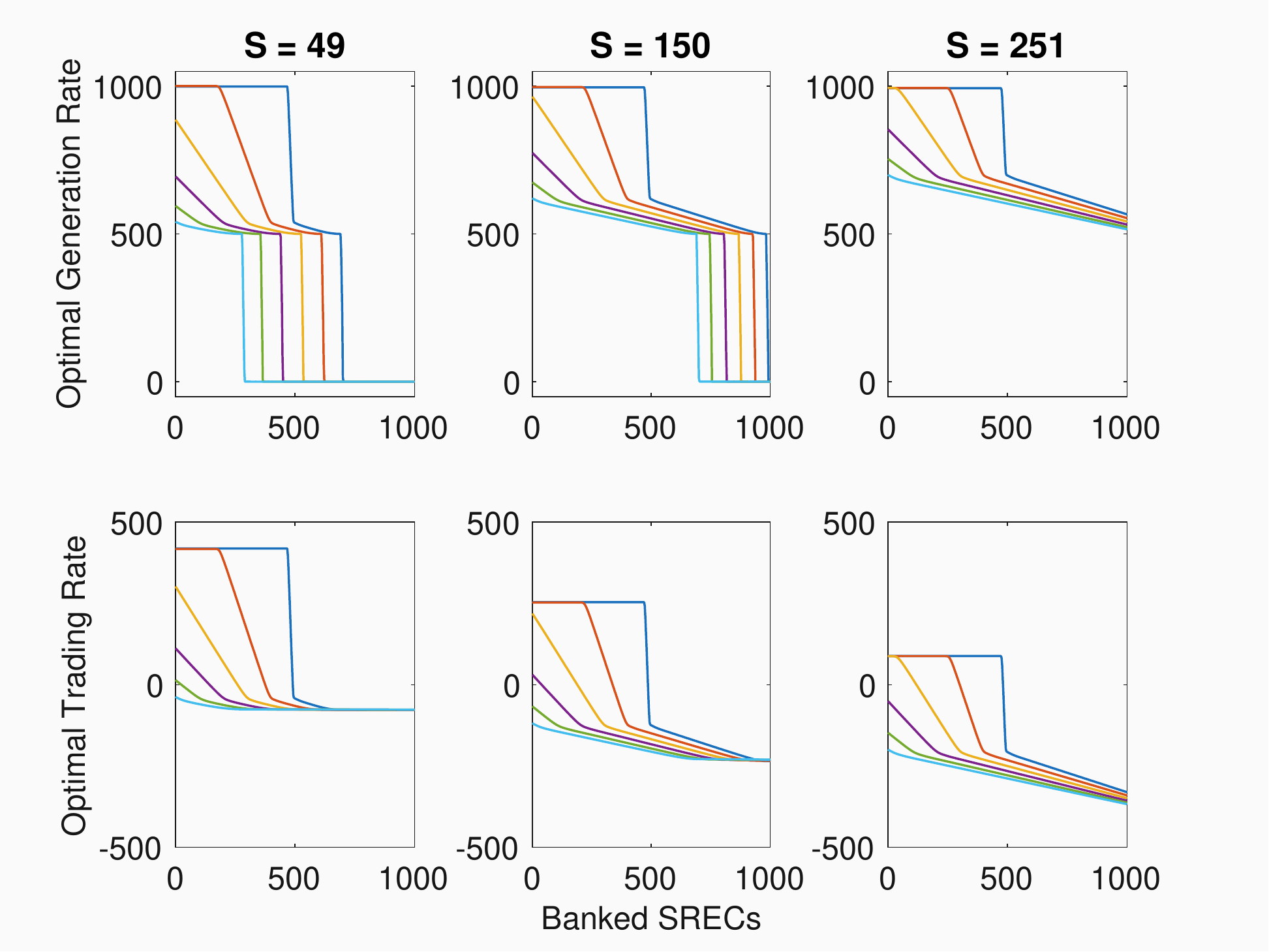}
\caption{Optimal firm behaviour as a function of banked SRECs across various time-steps (during the second and third of five compliance periods) and SREC market prices. Parameters as in Tables as in Tables \ref{tbl:ComplianceParams} and \ref{tbl:ModelParams}.}
\label{fig:multi-period-23_pi}
\end{figure}
\begin{figure}[h]
\centering
{\large$\boldsymbol{m = 4, 5}$\qquad\qquad\quad}
\\
\includegraphics[align = c, width=0.4\textwidth]{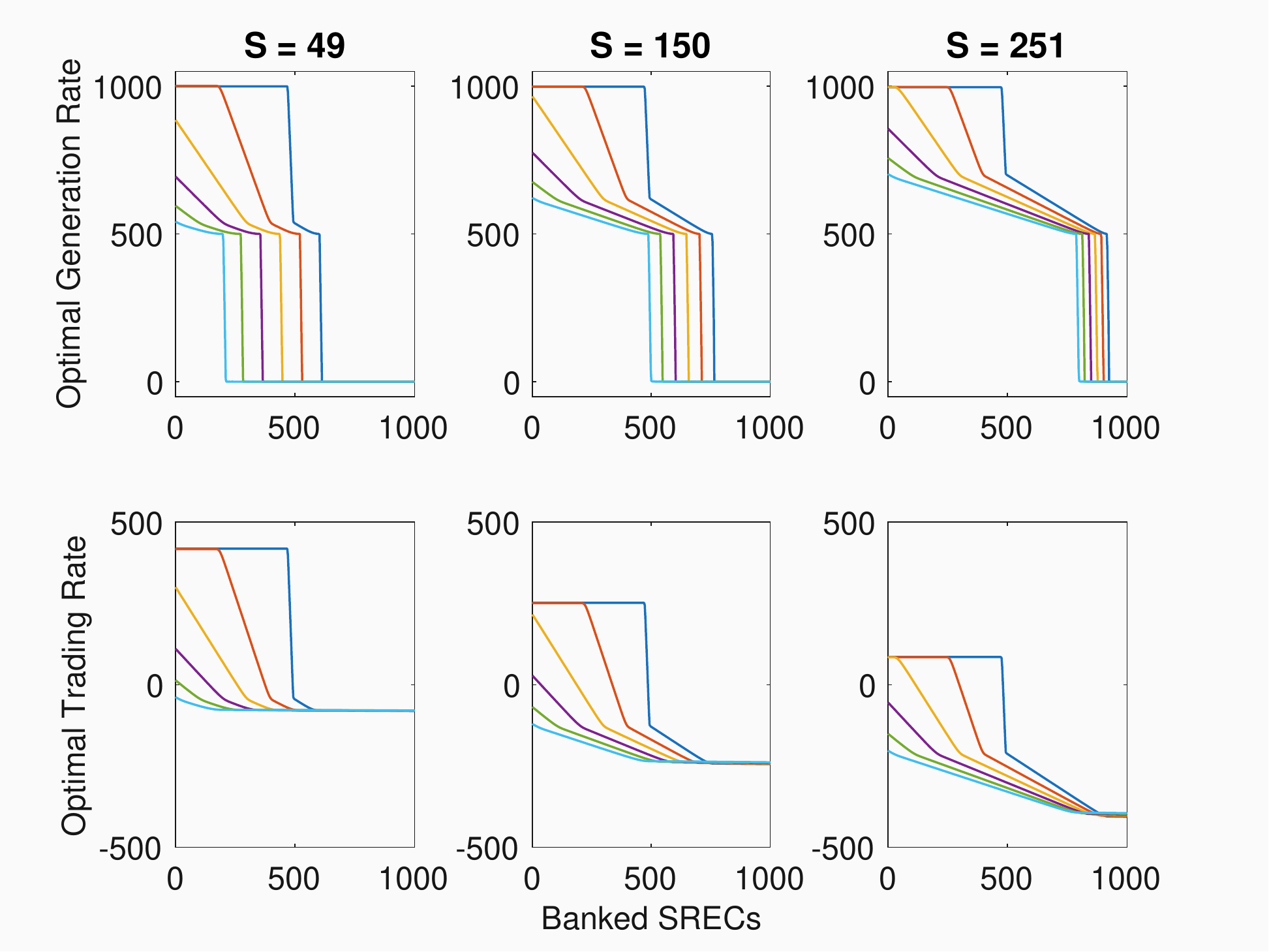}
\includegraphics[align = c, width=0.4\textwidth]{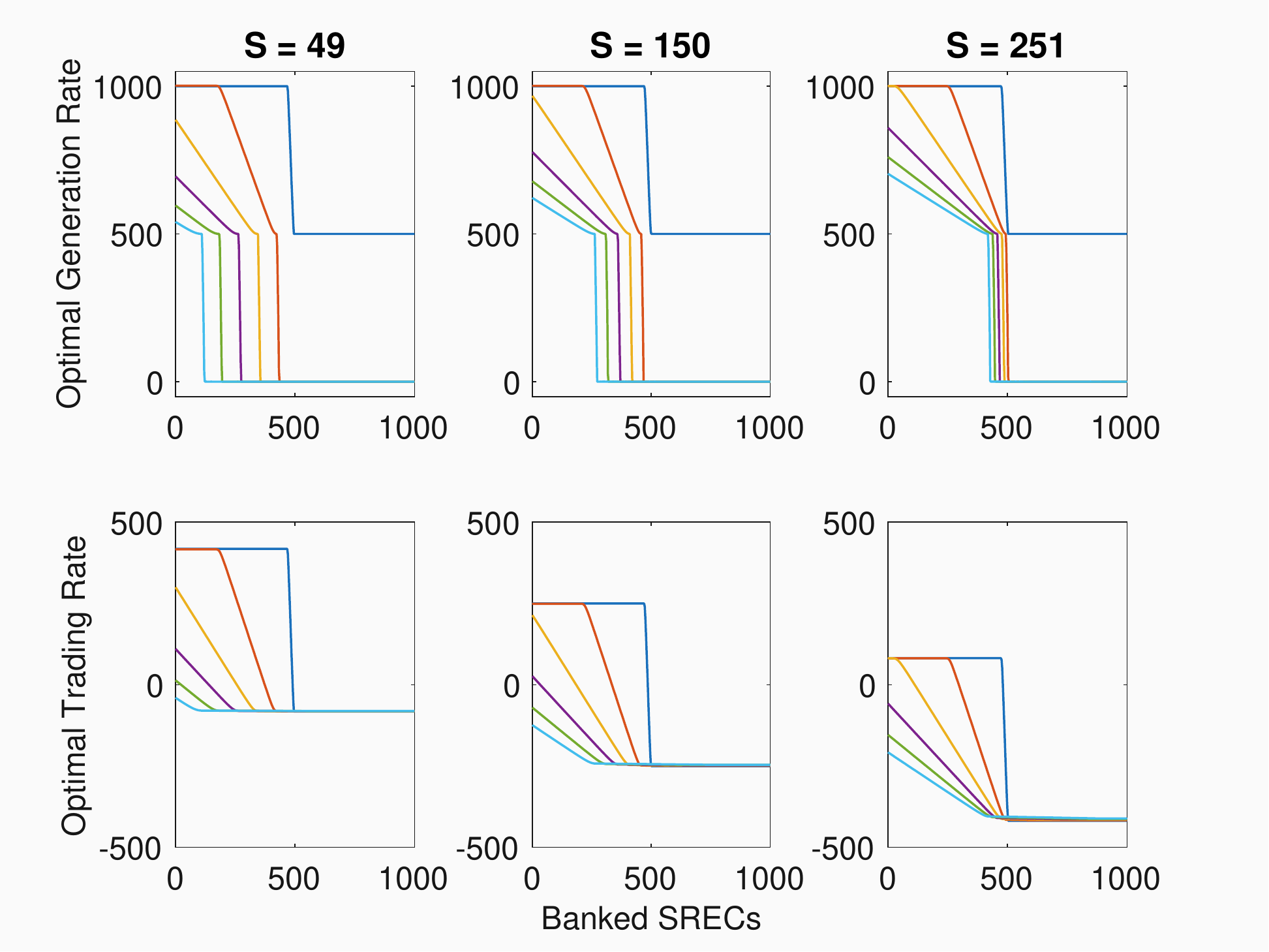}
\caption{Optimal firm behaviour as a function of banked SRECs across various time-steps (during the fourth and fifth of five compliance periods) and SREC market prices. Parameters as in Tables \ref{tbl:ComplianceParams} and \ref{tbl:ModelParams}.}
\label{fig:multi-period-45_pi}
\end{figure}


\end{document}